\documentclass[aps,10pt,twocolumn,showpacs,prx,citeautoscript,amsmath,amssymb,floatfix,nofootinbib,superscriptaddress,longbibliography]{revtex4-1}
\usepackage[left=2cm,right=2cm,top=2cm,bottom=2cm]{geometry}

\usepackage[utf8]{inputenc}
\usepackage[english]{babel}
\usepackage[T1]{fontenc}
\usepackage{amsmath}
\usepackage{amsfonts}
\usepackage{dsfont}
\usepackage{amssymb}

\usepackage{enumerate}
\usepackage{braket}
\usepackage{framed}
\usepackage{graphicx}

\usepackage{hyperref}
\usepackage{cleveref}
\usepackage{natbib}
\usepackage{csquotes}
\usepackage{xcolor}
\usepackage[normalem]{ulem}

\usepackage{bbold}

\DeclareMathOperator{\Tr}{Tr}

\usepackage{amsthm}
\theoremstyle{plain}
\newtheorem{theorem}{Theorem}
\newtheorem{lemma}[theorem]{Lemma}
\newtheorem{corollary}[theorem]{Corollary}
\newtheorem{definition}{Definition}

\newtheorem{remark}{Remark}
\newtheorem*{remark*}{Remark}
\newtheorem*{corollary*}{Corollary}
\newtheorem*{theorem*}{Theorem}
\newtheorem{proposition}{Proposition}

\usepackage{soul}

\begin{document}
\title{Emergence of noncontextuality under quantum Darwinism}
\author{Roberto D. Baldij\~{a}o}
\email[Corresponding author: ]{rdbaldi@ifi.unicamp.br}
\affiliation{Institute of Physics `Gleb Wataghin', University of Campinas--UNICAMP,
13083-859, Campinas, SP, Brazil}
\author{Rafael Wagner}
\affiliation{International Iberian Nanotechnology Laboratory (INL) Av. Mestre José Veiga, 4715-330 Braga, Portugal}
\affiliation{Department of Mathematical Physics, Institute of Physics, University of S\~ao Paulo, R. do Mat\~ao 1371, S\~ao Paulo 05508-090, SP, Brazil}
\affiliation{Centro de F\'{i}sica, Universidade do Minho, Braga 4710-057, Portugal}
\author{Cristhiano Duarte}
\affiliation{Wigner Research Centre for Physics, H-1121, Budapest Hungary}
\affiliation{School of Physics and Astronomy, University of Leeds, Leeds LS2 9JT, United Kingdom}
\affiliation{International Institute of Physics, Federal University of Rio Grande do Norte, 59070-405 Natal, Brazil}
\author{B\'{a}rbara Amaral}
\affiliation{Department of Mathematical Physics, Institute of Physics, University of S\~ao Paulo, R. do Mat\~ao 1371, S\~ao Paulo 05508-090, SP, Brazil}
\author{Marcelo Terra Cunha}
\affiliation{Universidade Estadual de Campinas, Cidade Universit\'aria Zeferino Vaz, 651, S\'ergio Buarque de Holanda, Campinas, SP,13083059, Brazil}

\date{\today}

\begin{abstract}
    Quantum Darwinism proposes that the proliferation of redundant information plays a major role in the emergence of objectivity out of the quantum world. 
    Is this kind of objectivity necessarily classical?  We show that if one takes Spekkens' notion of noncontextuality as the notion of classicality and the approach of Brand\~{a}o, Piani and Horodecki to quantum Darwinism, the answer to the above question is `yes', if the environment encodes sufficiently well the proliferated information. Moreover, we propose a threshold on this encoding, above which one can unambiguously say that classical objectivity has emerged under quantum Darwinism.
\end{abstract}
\maketitle

\section{Introduction}
Quantum theory is a normative set of rules whose explanatory power finds no counterpart in classical theories~\cite{Bell66,Bell64,JL16}. As a framework adapted to deal with probabilistic descriptions of microscopical phenomena, quantum theory generalises, and includes, the whole set of classical probability theories~\cite{GohEtAl18}. Remarkably, quantum theory's generalisation and explanatory capacity come together with a long list of extremely useful practical applications~\cite{BCPSW14} and philosophic implications which challenge our classical understanding of physical reality~\cite{PBR12,Leifer14}. 

Amongst the many challenges that quantum theory presents, two of them still merit attention and a proper connection. First, an unambiguous definition of what is classical and what is quantum. Second, the emergence out of a deeper microscopical quantum world of the objective reality we are familiar with. Arguably, an important candidate for the former is \emph{noncontextuality}, as it expresses an important ingredient of the reasoning underlying classical probability theories  -- operational equivalences reflect and imply ontological equivalences~\cite{spekkens2019_Leibniz,Spekkens_2008_C&Negativity,Spekkens_2005_C,schmid2020_StructureTheorem}. 

To investigate the emergence of an objective {(commonly agreed upon) reality,} we  consider that the environment plays a crucial role in mediating  interactions between observers and the observed system. Within this paradigm, \emph{quantum Darwinism} is undoubtedly one of the most physically appealing {notions, seeking to explain why distinct and independent observers usually}  obtain the same information regarding the system they are interacting with~\cite{Zurek_2003_Review_QD0, Zurek_1981_PointerStates,zurek_2007_RelativeQDReview}. In other words, quantum Darwinism helps to explain emergence of objectivity. Here we consider the approach to quantum Darwinism due to Brand\~{a}o, Piani and Horodecki ~\cite{BPH_2015_QD}, which is independent of specific aspects of the system-environment interaction, thus being capable to show generic emergence of objectivity. {Moreover, it encapsulates some essential features of other important processes related to emergence of objectivity under the quantum realm~\cite{Zurek_2003_Review_QD0,Horodecki_SSB_2015,Le_StrongQD_2019}}.

{ 
Can we say that the objectivity obtained in Darwinist processes necessarily implies classicality?} { Considering solely  Brand\~{a}o, Piani and Horodecki}{'s original argument}, this is certainly not clear: as we shall see, no restriction is made on the information upon which {the} {observers agree. Moreover, a recent result reinforce the idea that agreement --classical or not-- may be an important defining aspect of \emph{quantum} theory itself, making this question even more subtle~\cite{ContrerasTejada2021_Agreement}}.
In this work we show that whenever {{`enough} objectivity {'} emerges} due to a Darwinist process, observers can construct noncontextual models explaining the statistics of their experiments. {Therefore,} the novel aspect we {put forward is that noncontextuality emerges out of {{(sufficiently successful)}} quantum Darwinism or, in other words, that {a precise notion of} classicality emerges from { this} notion of objectivity.} { Moreover,} we discuss that even if quantum Darwinism does not hold, contextuality might still fade out in an \textit{environment as a witness} dynamics, thus being a more sensitive feature. { Finally, we propose a quantitative condition for emergence of \emph{classical objectivity} under quantum Darwinism}. Surprisingly, our proofs can also be of interest to quantum state discrimination, as we show that a high success probability on this task can point to a geometric structure of the states being discriminated - they must be vertices of a simplex in the space of density matrices.

\section{Background}

\subsection{Environment as a witness and Quantum Darwinism}

{{
Quantum Darwinism aims to explain how independent observers may obtain the same information about a central quantum system, despite the odd features of quantum theory that defy a classical notion of objectivity. For instance, if quantum theory may be applied to arbitrary systems and a cat's probability of staying alive or dying are both non-null, why do we all see the same outcome (hopefully alive)? Ultimately, providing an answer to this question explains the emergence of an objective world out of the quantum realm -- where the notion of objectivity comes from the common experience of potential independent observers\footnote{{{Given the dependence on the observers' experience, this notion has also been called `inter-subjectivity' rather than objectivity \cite{Mironowicz2017_InterSubj}. Nevertheless, we will use the term `objectivity', as often used in the quantum Darwinism literature.}}}. In this section we concisely describe the results of \citet*{BPH_2015_QD} that will be of importance to this work.}
}

{{Initially proposed in ref.~\cite{Zurek_2003_Review_QD0}, quantum Darwinism  recognizes the active role of a fragmented environment in the emergence of objectivity, as dictated by the so-called `Environment as a Witness' paradigm. }} Broadly speaking, potential observers extract information about the system of interest \emph{ by interacting with a {portion} of its environment} -- rather than interacting with the system itself~\cite{Ollivier_2005_EW}. 

Let us look into the Environment as a Witness paradigm in more detail (see fig \ref{fig:Paradigms}). The central system of interest, $A$, interacts with its environment, $\mathcal{E}$,  {composed} of {$N$} subsystems, $B_1,\ldots,B_{N}$.
We focus on the information that the environment has about $A$ represented by a {completely positive and trace-preserving (cptp)} map $\Phi: \mathcal{D}(\mathcal{H}_A)\rightarrow \mathcal{D}(\mathcal{H_E})$, where $\mathcal{D}(\mathcal{H}_A)$ represents the density matrices acting on the central system's Hilbert space $\mathcal{H}_A$, and analogously for $\mathcal{D}(\mathcal{H_E})$. Now, the actual dynamics of interest, which we {name} EW$_t$-dynamics { (short for Environment as a Witness $t$-dynamics)}, considers only a portion of such environment, and is defined as follows. 

\begin{definition}[${\rm EW}_t$-dynamics]
\label{def:EWdynamics}
Consider a {dynamics} 
$\Phi: \mathcal{D}(\mathcal{H}_A)\rightarrow \mathcal{D}(\mathcal{H_E})$, from a central system to a multi-partite environment, $A$ and $\mathcal{E}$ (resp.).
{Let ${S_{t} \subset \{1,...,N\}}$ be a set of labels describing $t$ portions of the environment $\mathcal{E}$. An {\textbf{EW$_t$-dynamics}} for the subset $B_{S_t}:=\{B_j\}_{j\in S_t}$ is the cptp map $\Phi^{B_{S_t}}:= \mathrm{Tr}_{\mathcal{E}\setminus B_{S_t}}\circ\Phi$}.
\end{definition}

{Tracing out $\mathcal{E}{\setminus}B_{S_t}$, an EW$_t$ dynamics selects exactly which fractions of the environment one wants to centre attention at.} \citet*{BPH_2015_QD} show that this dynamics {already implies} a certain kind of objectivity, called objectivity of observables,  in which a    `pointer POVM'~\cite{Beny_2008_unsharp} is selected by the interaction.
\begin{theorem}[Adapted Thm. $2$ of \cite{BPH_2015_QD}, Objectivity of observables]
\label{theo: BPH1}
Consider an EW$_t$-dynamics $\Phi^{B_{S_t}}$. If the total environment is big enough compared to $S_t$ ($N\gg t$), then there exists a POVM $\{\tilde{E}_k\}$ acting on $\mathcal{D}(\mathcal{H}_{A})$ such that, for most choices of $S_t$,
\begin{equation}
    \Phi^{B_{S_t}}(\rho^A) \approx \sum_k {\rm Tr}[\tilde{E}_k \rho^A]\sigma^{B_{St}}_k,
    \label{eq: ObjectivityObs}
\end{equation}
where $\sigma^{B_{S_t}}_k\in\mathcal{D}(\bigotimes_{j\in S_t}\mathcal{H}_{B_j})$ { and $\{\tilde{E}_k\}_k$ is independent of $B_{S_t}$}.
\end{theorem}
{A consequence of this theorem is that, when its conditions are met, the environmental state $\Phi^{B_{S_t}}(\rho^A)$ approximately} encode{s} {information regarding the probability distribution $\tilde{p}_k:={\rm Tr}[\tilde{E}_k \rho^{A}]$. (Note, however, that nothing is said on how good this encoding is.}) The most important feature of this result is that the POVM $\{\tilde{E_k}\}_k$ \emph{does not depend on the elements of $B_{S_t}$}. {The emergence of an objective reality comes from this fact, i.e. that} the observable (possibly) monitored by the environment is essentially independent of which portion of $t$ subsystems of the environment one ends up with.

From now on, we only consider $S_t$ to be portions of the environment such that the approximation in {Eq.}\eqref{eq: ObjectivityObs} is valid. Moreover, we assume {an infinite-sized environment, as is usual in open quantum systems approaches, such that the approximation can be substituted by an equality (see sec.~\ref{App: BPHandRelax} of the Appendix for details and for a relaxation of this assumption). }
Without loss of generality, we also assume that {each} $\tilde{E_k}\neq 0$.

With {these} assumptions, we see that any EW$_{t}$-dynamics, $\Phi^{B_{S_t}}$, is effectively a measure-and-prepare map, defined by the pairs $(\tilde{E}_k,\sigma^{B_{S_t}}_k)_k$, where $\{\tilde{E}_k\}_k$ form{s} a POVM {acting} on the central system with non-null elements{,} and $\{\sigma^{B_{S_t}}_k\}_k$ are states of $B_{S_t}$ { (as such, these states may depend on which portion of the environment we are focusing on)}.

\begin{figure}[ht]
    \centering
    \includegraphics[width=\columnwidth]{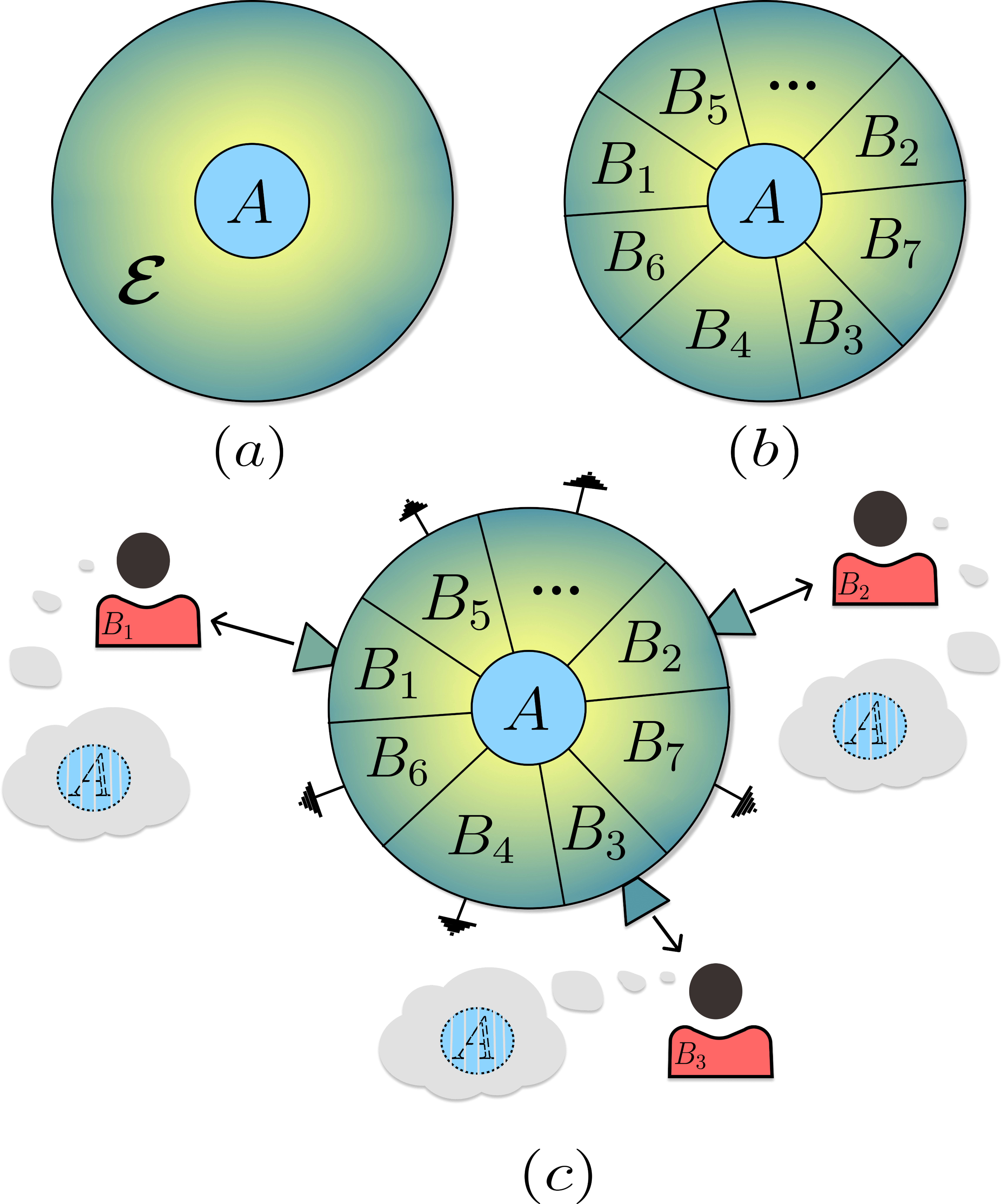}
    \caption{{{\textbf{Different approaches to open systems.} $(a)$ \textbf{Decoherence paradigm}: The focus is on the central system $A$. The environment is treated as a monolithic system, usually to be traced out after the interaction. $(b)$ \textbf{Environment split into fractions}:  The environment is more accurately described as  a composition of several subsystems.
    $(c)$ \textbf{Environment as a witness dynamics and quantum Darwinism process}: The environment under this paradigm is promoted to an information channel, spreading redundant information about system $A$ to independent observers. { In the example of the figure, the number of portions in $S_t$ is $t=3$  and $B_{S_t}=\{B_1,B_2,B_3\}$.}  Theorem \ref{theorem: NCunderQD} {tells us a sufficient condition} for such redundant information to be noncontextual, i.e., classical.}}}
    \label{fig:Paradigms}
\end{figure}
The objectivity of observables obtained through EW$_t$-dynamics tells that different observers might gather information regarding the same observable, $\{\tilde{E}_k\}_k$, (almost) independently of the choice of $B_{S_t}\subset\mathcal{E}$. However, objectivity of observables does not solve adequately the emergence of an objective reality, as a crucial feature of {the chosen notion of} objectivity is that independent observers should agree on which \emph{outcome} they see{. In other words, they should all be able to infer the same label $k$ upon measuring their share of the environment} -- and this is not a consequence of objectivity of observables. \citet*{BPH_2015_QD} tell us that this happens when the full Darwinist process occurs, i.e., when information about $\{\tilde{E}_k\}_k$ is indeed well-encoded on the environmental subsystems. 

{Let us make the above claim more precise.} Consider a setting with several observers (call them Bob). For simplicity, assume we have $t$ Bobs, each receiving one subsystem of $B_{S_t}$, and let us label each Bob by the element $B_j\in B_{S_t}$. Each Bob will be the end of a reduced dynamics, $\Phi^{B_j}:={\rm Tr}_{B_{S_t}\setminus B_j}\circ\Phi^{B_{S_t}}$, so Bob $B_j$ will hold a state of the form $\sum_k\tilde{p}_k \sigma^{B_j}_k$, where $\sigma^{B_j}_k=\text{Tr}_{B_{S_t}\setminus B_j}[\sigma^{B_{S_t}}_k]$.
Now, the hypotheses that different subsystems encode well enough the information regarding the selected observable, as presumed by quantum Darwinism \cite{Zwolak_2017_IrrelevanceofIrrelevant},
brings an important requirement in the landscape of reference \cite{BPH_2015_QD}: 
sufficient distinguishability of the states $(\sigma^{B_j}_k)_k$ for each $B_j$. We thus define
\begin{definition}[QD$_\eta$ process]
Consider an EW$_{t}$-dynamics $\Phi^{B_{S_t}}:=(\tilde{E}_k,\sigma^{B_{S_t}}_k)_k$ and the $t$ partial traces $\Phi^{B_j}$. 
Define the {quantity}
\begin{align}
    p_{\rm guess}[(\tilde{p}_k,\sigma^{B_j}_k)_k]:=\max_{\{F_k^{B_j}\}}\sum_k \tilde{p}_k{\rm Tr}\left[  F^{B_j}_k\sigma^{B_j}_k\right], 
\end{align}
where $\tilde{p}_k={\rm Tr}[\tilde{E}_k\rho^A]$, $\sigma^{B_j}_k = {\rm Tr}_{B_{S_t}\setminus B_j}[\sigma^{B_{S_t}}_k]$ and $\{F^{B_j}_k\}_k$ are local POVMs for each Bob. {A} {\textbf{Darwinism process with distinguishability}} $\eta$, QD$_\eta$, is said to occur {when} { there exists an $0 \leq \eta \leq 1$ such that}, for all $B_j\in B_{S_t}$,  
\begin{align}
    \min_{\rho^A} p_{\rm guess}[(\tilde{p}_k,\sigma^{B_j}_k)]\geq\eta.
    \label{eq:QDeta}
\end{align}
\end{definition}

The quantity $p_{\rm guess}$ tells us how well one can discriminate between states $\{\sigma_k^{B_j}\}_k$, depending on $\rho^A$ {(recall, $\tilde{p}_{k}$ is a function of $\rho_{A}$}). Thus, QD$_\eta$ sets a lower bound on the distinguishability of the $\{\sigma_k^{B_j}\}_k$, \emph{to all $B_j$} { and regardless of the initial state of the central system}. From now on, when we write QD$_\eta$ we assume that $\eta$ is the maximal value for which Inequality \eqref{eq:QDeta} holds.

Finally,  \citet*{BPH_2015_QD} tell us that a Darwinist process essentially leads to { all Bobs seeing the same outcome}:
\begin{proposition}[Adapted { Prop. $3$} of \cite{BPH_2015_QD} {, Objectivity of Outcomes}]
\label{remark: QDObjectivityofOutcomes}
{ Consider the probability that each Bob, by performing a local measurement $\{F^{B_j}_k\}_k$, get the same outcome $k$, provided that the encoded label was indeed $k$,i.e.:
\begin{align}
    \min_{\rho^A}\sum_k\tilde{p}_k{\rm Tr}\left[\bigotimes_{j\in{S_t}}F_k^{B_j}\sigma^{S_t}_k\right].
\end{align}
}
{{ For any QD$_\eta$ process,}} there exists local POVMs {$\{\bar{F}_k^{B_j}\}_k$} {for each Bob} such that
\begin{align}
    \min_{\rho^A} \sum_k \tilde{p}_k {\rm Tr}\left[\bigotimes_{{j\in S_t}} {\bar{F}}_k^{B_j}\sigma^{S_t}{_k}\right]\geq 1- 6t\delta^{\frac{1}{4}},
\end{align}
where $\delta=1-\eta$.
\end{proposition}

{{In other words, the better the distribution $Tr[\tilde{E}_k\rho^A]$ is encoded on the states $\{\sigma_k^{B_j}\}$, the higher the probability that all Bobs will see the same outcome (provided they implement the correct measurements). Note that, whether or not the Bobs will see the same outcome is not the key aspect here; rather, what is important is that if they try to obtain the information correctly, they will agree. Roughly speaking, this means that the higher the $\eta$ in a QD$_\eta$ process, the higher the `amount' of emergent objectivity. We thus have a continuous emergence of objectivity under this approach to quantum Darwinism. In the (often considered) limit $\eta\rightarrow1$ of perfect encoding, nearly perfect objectivity is obtained.}}
{Whether this condition of `high $\eta$' is generic or not among physically motivated interactions is still an open question. Even though it is a natural assumption to take in the Environment as a Witness paradigm, knowing how generic or which interactions will have this feature is vital for the Darwinism program, and we hope that future work will provide such an answer. 

The emergence of objectivity under quantum Darwinism in the framework of { Brand\~{a}o, Piani and Horodecki} is a {remarkable} result. 
}
It comes, nevertheless, with a small drawback{, as} there are no restrictions on {the nature of} $\{\tilde{E}_k\}_k$~\cite{BPH_2015_QD}{-- in principle, $\{\tilde{E}_k\}_k$ can even represent an informationally complete POVM}.
This is incompatible with the idea of a `classical observable'  selected by the interaction,  {as usually thought of within} the Darwinist program \cite{Zwolak_2017_IrrelevanceofIrrelevant}. {The possibility of having {informationally complete} POVMs spreading over the} environment means that essentially \emph{all information} ({including} \emph{quantum} information) contained in $\rho^A$ can be stored in the environment.

Can we still say that { some sort of classicality}  arise in those cases?
Here we {argue} that, even  if we allow for $\{\tilde{E}_k\}_k$ to be an {{informationally complete POVM}}, there is still {a precise} notion of emergence of classicality {in} { some} QD$_\eta$ processes: all the Bobs can construct a noncontextual ontological model for the statistics of their measurements on every $\Phi^{B_j}(\rho^A)$.

\subsection{Spekkens Contextuality}

The conception of classicality we adopt is {the} Spekkens' notion of noncontextuality~\cite{Spekkens_2005_C} {for} prepare-and-measure scenarios. {This notion {captures} essential features of classical theories, and thus encapsulates several criteria associated to classicality in different research areas. For instance, {Spekkens's} noncontextuality is equivalent to the existence of positive quasiprobability distributions~\cite{Spekkens_2008_C&Negativity}, which is a widely used classicality criteria in quantum optics~\cite{mandel_wolf_1995,Wigner_Quasiprobability}. In the framework of generalized probabilistic theories~\cite{Barrett_GPT,Janotta_GPT}, in which physical principles can be studied without resorting to the quantum or classical formalisms, a noncontextual theory must be embedded in classical probability theory ~\cite{schmid2021characterization}.
Moreover, {Spekkens's} noncontextuality also encapsulates }{ Gaussian quantum mechanics ~\cite{bartlett2012reconstruction}, and subsumes the well-known notions of Kochen-Specker noncontextuality \cite{KunjwalSpekkens_KSwithoutDeterminism,kunjwal2018statistical, Kunjwal_2019,Kunjwal_2020} and local causality \cite{schmid2018all}.  In this section, we briefly describe this idea of classicality.}

{Spekkens' framework is operational and allows for a representation of lab procedures -  a list $\mathcal{P}$ of preparations, a list $\mathcal{M}$ of measurements, and a set of effects $\{b|M\}_{b}$ for each measurement $M$ in $\mathcal{M}$}.
In particular, {it is assumed that convex combinations of preparations in $\mathcal{P}$ also result in valid preparations.}
In general, there will be equivalences among procedures of the same type: a preparation $P_1\in\mathcal{P}$ is equivalent to $P_2\in \mathcal{P}$ in a prepare-and-measure scenario when, for all \emph{conceivable} measurement procedures, $p(b|M,P_1)=p(b|M,P_2)$. That is, one cannot differ those {preparation} procedures from the probabilities arising on \emph{any} measurement -- or, equivalently for a subset of measurements called tomographically complete. The above relation defines equivalence classes of preparations, each class denoted by $[P]$. For instance, in quantum theory, a density matrix could label an equivalence class: all the alternative ways to produce a certain $\rho\in\mathcal{D}(\mathcal{H})$ will always lead to the same statistics. Similarly, we can define the equivalences for events of measurements: they are equivalent when no \emph{conceivable} preparation leads to different probabilities. In quantum theory, these can be labeled by the elements of POVMs. We denote the sets of all equivalences of each type of procedure as ${\rm Equiv}\{\mathcal{P}\}$ and ${\rm Equiv}\{\mathcal{M}\}$. 

As an attempt to {classically} explain how the probabilities arise within the operational theory, one might try to construct an ontological model. Ontological models in prepare-and-measure scenarios comprise of three ingredients: a measurable space $(\Lambda,\Sigma)$; a linear map $\mu: \mathcal{P}\mapsto (\mu_P)_{P\in \mathcal{P}}$, where $\mu_P(\lambda)$ is a probability measure over $\Lambda$; and another linear map $\xi: (b|M) \mapsto \xi_M(b|\lambda)$, such that each measurement $M\in\mathcal{M}$ is mapped to a set of response functions $(\xi_M(b|\lambda))_b$ satisfying $\xi_M(b|\lambda)\geq0$, $\sum_b \xi_M(b|\lambda)=1\,\, \forall \lambda$.
Linearity of the maps defined above implies that convex mixtures of procedures should be mapped to the appropriate convex mixtures of the corresponding ontic objects. Finally, the probabilities of the operational theory must be reproduced by ontological models via $p(b|M,P)=\sum_\lambda \mu_P(\lambda)\xi_M(b|\lambda)${~\cite{Leifer14}}. Up to here, any operational theory will accept explanations via some ontological model. However, the important aspect is { whether it is possible} to recover  operational probabilities via a \emph{noncontextual ontological model}~\cite{Spekkens_2005_C}:
\begin{definition}[Noncontextual ontological models] Consider an operational theory with procedures $\mathcal{P,M}$ and equivalences ${\rm Equiv}\{\mathcal{P}\},{\rm Equiv}\{\mathcal{M}\}$. An ontological model recovering the probabilities of the operational theory is said to be \textbf{preparation noncontextual} if:
\begin{align}
    P\in[P']\implies \mu_P = \mu_{P'}.
\end{align}
Analogously, the model is said to be \textbf{measurement noncontextual} if
\begin{align}
    b|M\in[b'|M']\implies \xi_M(b|\lambda) = \xi_{M'}(b'|\lambda)\,\,\forall \lambda\in\Lambda.
\end{align}
If both conditions are satisfied, the ontological model is said to be \textbf{universally noncontextual}, or noncontextual, for brief.
\end{definition}

{This definition says that} procedures that are operationally indistinguishable should be represented by the same entities on the ontological model. {The essence of this definition of noncontextuality is the so-called principle of the identity of indiscernibles, due to Leibniz \cite{Spekkens_2005_C, spekkens2019_Leibniz}. Roughly, it says that two empirically indistinguishable things should be ontologically described as \emph{identical}. The methodology implied by (generalizations of) this principle was important in building classical theories in the past; for instance, it played an important role in the birth of relativity\cite{spekkens2019_Leibniz}. } 

{ Even though noncontextuality is a natural feature to {demand} of physical theories,} quantum theory is known to be a `contextual' operational theory: there is no noncontextual ontological model {reproducing} { all of its statistics ~\cite{KS,Spekkens_2005_C}}.
{ This remarkable nonclassical feature of quantum theory has important practical consequences: contextuality is a useful resource, as it provides genuine quantum advantage over classical computation.} {{ Indeed, contextuality is {the} {crucial} resource in information-processing~\cite{spekkens2009preparation,saha2019preparation,ghorai2018optimal,ambainis2019parity,saha2019state,chailloux2016optimal}, and provides advantage for quantum state discrimination~\cite{schmid2018contextual}, quantum cloning~\cite{lostaglio2020contextual}, quantum metrology~\cite{lostaglio2020certifying} and quantum machine learning~\cite{gao2021enhancing}.  It also constitutes a necessary resource for models of restricted quantum computation~\cite{schmid2021only,howard2014contextuality,raussendorf2013contextuality,mansfield2018quantum}, capable of providing robust numerical bounds of advantage~\cite{abramsky2017contextual}, not only for deterministic but also for probabilistic restricted quantum computation. }}

{In summary, Spekkens{'s} notion of noncontextuality is a { mathematically well-defined} and { amply used} notion of classicality.  { In this sense, it is fair to assume that the} emergence of generalized noncontextuality in an intrinsically contextual theory {can be} naturally understood as a classical limit { -- with a very precise meaning}. In what comes, we  show a deep connection between classicality in these terms and Quantum Darwinism objectivity}.

\section{Results}

\subsection{Emergence of noncontextuality}

{T}o approach the question of emergence of noncontextuality{, as a proxy for { emergence of} classicality} in the considered dynamics, we need to  define what is the scenario we care about -- we need to say what are the procedures and operational equivalences, thus establishing what a noncontextual model must 
respect. As {demanded} by the { Environment as a Witness} paradigm, we {focus on} Bob's perspectives:  \emph{to each Bob}, we must say what the different preparations, measurements and equivalences {are}. {Measurement procedures are straightforward}: each Bob implements a measurement procedure $M^{B_j}\in\mathcal{M}^{B_j}$ on the system he receives. Therefore, the equivalence classes are {identified with} the elements of POVM at each Bob's side.

A more subtle subject in our case is defining preparations and equivalences thereof. Suppose, for sake of simplicity, that there is an { experimenter} (say, Alice)   preparing {the} {quantum states of system} $A$, according to a set of possible preparations $\mathcal{P}$; then, the equivalence class of $P\in\mathcal{P}$ would be $[\rho_P^A]$  ({if there \emph{really is} an { experimenter} is not of importance to this work. The procedures can be understood as any specific sequence of events ``preparing''  system $A$}).    
However, each Bob will \emph{infer} information about system $A$ \emph{by making measurements on their share of $A$'s environment}. Therefore, {from} Bob $B_j$'s perspective, every preparation of system $A$ is always followed by a fixed transformation $T_j$, which we can describe as: interact $A$ with a {large} {number} of systems in a predefined way and discard all subsystems but $B_j\in B_{S_t}$. Thus, the preparations arriving {at} $B_j$ are effectively $\mathcal{P}'_j:=T_j(\mathcal{P})$ (see {fig.} \ref{fig:ScenarioBj}).
{From Bob's standpoint,} the equivalence classes for preparations must then be defined by the density matrix {available for Bob}, i.e. $\Phi^{B_j}(\rho_P^A)$. In sum, the scenario of each $B_j$ is {composed of} preparations $\mathcal{P}_j'$ with equivalences $[\Phi^{B_j}(\rho^A)]$, while measurement procedures are the local implementations $\mathcal{M}^{B_j}$, with equivalence classes for effects represented by elements of POVMs.
\begin{figure}[ht]
    \centering
    \includegraphics[width=\columnwidth]{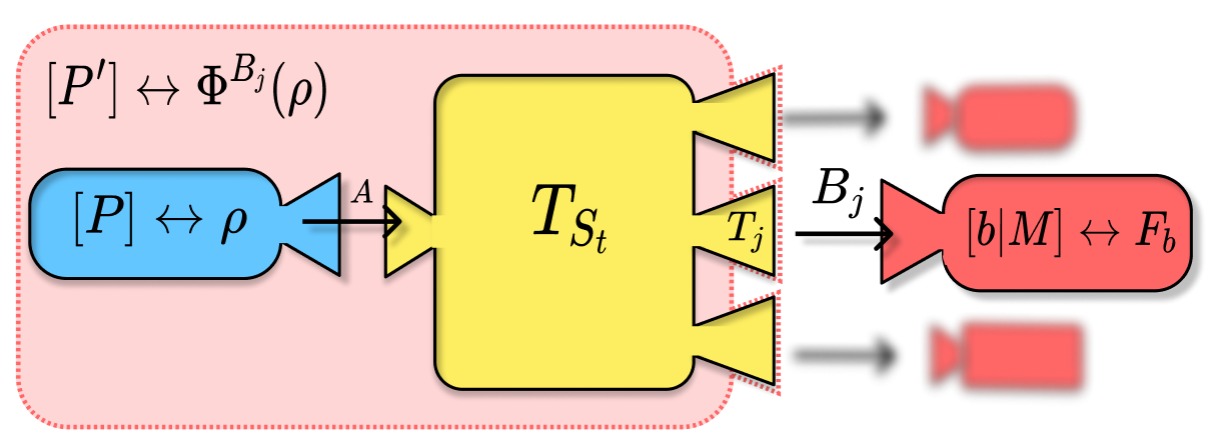}
    \caption{\textbf{Scenario from Bob $B_j$'s perspective}: Preparation $P\in\mathcal{P}$ of the central system is followed by \emph{a fixed} transformation (i.e., interaction with the environment, always prepared in the same way){, leading to an effective preparation $P'=T_j(P)$}. The non-traced out portion of the environment is given by the labels $S_t$ (here $t=3$). For Bob, $[P']\leftrightarrow \Phi^{B_j}(\rho^A)$ and $[b|M]\leftrightarrow F_b$.}
    \label{fig:ScenarioBj}
\end{figure}

Now we have set the stage to present our results. Bellow, we show a mathematical condition on the encoding states $\{\sigma^{B_j}_k\}_k$ of an EW$_{t}$ dynamics that, if satisfied for some $B_j\in S_t$, ensures that there exists a noncontextual ontological model recovering the statistics for those Bobs.
\begin{lemma}[Affinely independent $\{\sigma^{B_j}_k\}_k$ {implies} { noncontextuality}]
\label{lemma: AIimplyNC}
Consider an EW$_t$-dynamics, {with an infinite environment}, defined by the pairs $(\tilde{E}_k,\sigma^{S_t}_k)_k$. If, for a given $B_j\in S_t$, the environmental systems $\{\sigma^{B_j}_k\}_k$ form an affinely independent set of matrices on $\mathcal{D}(\mathcal{H}_{B_j})$, then there exists a noncontextual ontological model for the distribution about $A$ available to $B_j$. {This holds true for any measurement $\{F^{B_j}_b\}$ implemented by  Bob $B_j$}. 

\begin{proof}
We prove this result by constructing such a noncontextual ontological model. As mentioned above, the equivalence classes for the `preparations' arriving to $B_j$ are defined by the quantum states $\sigma^{B_j}_{P'}=\sum{\rm Tr}[\tilde{E}_k\rho_P^A]\sigma^{B_j}_k$. Affine independence of $\{\sigma^{B_j}_k\}_k$ implies the bijection
\begin{align}
\label{eq:bijectionFromAI}
    \sigma^{B_j}_{P'_1}=\sigma^{B_j}_{P'_2} \iff {\rm Tr}[\tilde{E}_k\rho^A_{{P}_1}] = {\rm Tr}[\tilde{E}_k\rho^A_{{P}_2}]\,\,\forall k.
\end{align}
(Such a bijection can also be seen by the fact that  affine independence of $\{\sigma^{B_j}_k\}_k$ imply that ${\rm ConvHull}\left[\{\sigma^{B_j}_k\}_k\right]$ is a $(k_{\rm max}-1)$-simplex on $\mathcal{D}(\mathcal{H}_{B_j})$; thus, every point has a unique convex decomposition in terms of the $k_{\rm max}$ vertices).
Therefore, the equivalences can be ultimately labeled by the distribution $({\rm Tr}[\tilde{E}_k\rho_P^A])_k$. Now, consider the following ontological model:
\begin{subequations}
\begin{align}
    \Lambda &:=\{k\}{_k};\\
    \mu_{P'}(k)&:={\rm Tr}[\tilde{E}_k\rho_P^A];\\
    \xi_M(b|k) &:= {\rm Tr}[F^{{B_j}}_b\sigma_k].
\end{align}
\label{eq: NCOMQD}
\end{subequations}
This indeed respects the equivalence classes of the scenario of $B_j$: $\mu_{P'}$ depends uniquely on the distribution ${\rm Tr}[\tilde{E}_k\rho_P^A]$, which labels each class $[P']$, while $\xi_M$ depends only on $F^{{B_j}}_b$, label of each $[b|M]$. It remains to show that this ontological model recovers the statistics. Since the probabilities are obtained via a measure-and-prepare channel defined by $(\tilde{E}_k,\sigma_k)_k$, we have $p(b|M,P')=\sum_k {\rm Tr}[\tilde{E}_k\rho_P^A]{\rm Tr}[F_b \sigma_k]=\sum_k \mu_{[P']}(k)\xi_{[M]}(b|k)$. 
\end{proof}
\end{lemma}

{
\begin{remark*}
The proof above has interesting connections to characterisations of noncontextuality in the framework of generalized probabilistic theories \cite{schmid2021characterization,shahandeh2021contextuality}. { In the Appendix \ref{app:NCGPT}, we make this connection more explicit for the interested reader.}
\end{remark*}
}

{Lemma~\ref{lemma: AIimplyNC}} shows how noncontextuality may \emph{emerge} in an {EW$_t$ dynamics : g}enerally, there is no noncontextual model for the scenario with full quantum probabilities, with { all possible} $\mathcal{P}$ and { their equivalences} ${\rm Equiv}\{\mathcal{P}\}$;
however, {the EW$_{t}$-}dynamics maps $[\rho^A]\mapsto[\Phi^{B_j}(\rho^A)]$,  {enabling} a noncontextual explanation.
{ Nonetheless, the condition of affine independence of $\{\sigma_k\}_k$ is not clearly motivated within quantum Darwinism. In other words, up until this point, Lemma \ref{lemma: AIimplyNC} is merely a mathematical condition signaling the existence of noncontextual models within EW$_t$-dynamics. Therefore, a connection between emergence of noncontextuality and emergence of objectivity under quantum Darwinism remains open. } 

{ By looking at the extreme case in which $p_{\rm guess}[(\sigma^{B_j}_k,\tilde{p}_k)_k]=1$ for some Bob, we find an indicative of such a connection: affine independence of $\{\sigma^{B_j}_k\}_k$ is necessary for perfect encoding of the distribution $({\rm Tr}\{\tilde{E}_k\rho^A\})_k$. Indeed}, if the distribution $({\rm Tr}\{{\tilde{E}}_k \rho^A\})_k$ is perfectly encoded {in} the states $\{\sigma^{B_j}_k\}_k$, these must be deterministically distinguishable; in turn, these states must have orthogonal support, thus being affinely independent. 
{
Moving away from the extreme case, w}e can get affine independence even if the discrimination is not perfect, but good enough. This is summarized on the following lemma {-- its} proof deferred to the Appendix \ref{App:lemma3}.
{Remarkably}, this proof can be of interest to quantum state discrimination: if the probability of success is high enough, we get an information about the geometrical disposition of $\{\sigma^{B_j}_k\}_k$, namely, that they form vertices of a ($k_{\rm max}$-1)-simplex in $\mathcal{D}(\mathcal{H}_{B_j})$.

\begin{lemma}[Operational signature implying affine independence]
\label{lemma: OperationalSignAI}
Consider an EW$_t$-dynamics defined by $({\tilde{E}}_k,\sigma_k^{B_j}){_k}$.
Then, there exists a bound $\hat{P}[({\tilde{E}}_k)_k]$ such that, if $p_{\rm guess}[(\tilde{p}_k,\sigma^{B_j}_k)]> \hat{P}[({\tilde{E}}_k)_k]\,\forall \rho^A$ {(recall that $\tilde{p}_{k}$ depends on the central system's state $\rho^{A}$)}, the set of states $\{\sigma^{B_j}_k\}_k$ must be affinely independent.
\end{lemma}
\begin{remark*}For the condition above to be valid, it is necessary that ${\rm dim}(\mathcal{H}_{B_j})\geq k_{\rm max}$.
\end{remark*}
In the Appendix (sec. \ref{App:lemma3}), we construct the bound $\hat{P}$, which is generally not trivial, i.e., there are cases where $\hat{P}<1$. Therefore, lemma \ref{lemma: OperationalSignAI} tells us that, if Bob $B_j$ can recover well enough the statistics of $({\rm Tr}\{{\tilde{E}_k\rho^A}\})_k$ 
through measurements on his system {(i.e., $p_{\rm guess} >\hat{P}$ for that Bob)}, $\{\sigma^{B_j}_k\}$ are affinely independent. Thus, by lemma \ref{lemma: AIimplyNC} he can construct a noncontextual model for such statistics. { Now,} if such good distinguishability also holds for every Bob $B_j\in S_t$, { a process of QD$_\eta$ occurs  with $\eta>\hat{P}$. In this case, while Brand\~{a}o, Piani and Horodecki's results \cite{BPH_2015_QD} put a lower bound on the probability that all Bobs will see the same outcome {{when making the correct measurements}},} the above results imply that \emph{every Bob can explain their results using a noncontextual model}. This is the summary of the following theorem, our main result:
\begin{theorem}[\textbf{Main result:} emergence of { noncontextuality} under QD$_\eta$]
\label{theorem: NCunderQD}
Suppose the conditions for QD$_\eta$ are met with $\eta>\hat{P}[({\tilde{E}}_k)_k]$. Then, each Bob can construct a noncontextual ontological model for the respective scenario.
\begin{proof}
If the conditions are met, it must hold that:
\begin{align}
    \min_{\rho} p_{\rm guess}[(\sigma_k^{B_j}),({\rm Tr}\{\tilde{E}_k\rho\})_k] >\hat{P}[(\tilde{E}_k)_k]\,\,\forall B_j\in S_t,
    \label{ineq:ViolationEveryBob}
\end{align}
which, in turn, imply
\begin{align}
    p_{\rm guess}[(\sigma_k^{B_j}),({\rm Tr}\{\tilde{E}_k\rho\})_k] >\hat{P}[(\tilde{E}_k)_k]\,\,&\forall \rho^A\in\mathcal{D}(\mathcal{H}_A)\nonumber\\
    &\text{ and } B_j\in S_t.
\end{align}

Therefore, by lemma \ref{lemma: OperationalSignAI} and Ineq. \eqref{ineq:ViolationEveryBob}, the states $\{\sigma^{B_j}_k\}$ must be affinely independent, \emph{for every} $B_j\in S_t$. Then, lemma \ref{lemma: AIimplyNC} guarantees that every Bob will be able to construct a noncontextual ontological model for their scenario.
\end{proof}
\end{theorem}

\subsection{A cut-off for classical objectivity in QD}
{

The emergence of objectivity in QD$_\eta$ processes has a continuous form: the closer $\eta$ is to $1$, {the} `more objectivity' arises due to Darwinism; only in the idealized limit $\eta=1$, ``perfect objectivity'' is obtained. { {This leads to problems, as there is no cut-off on $\eta$ defining if objectivity has emerged or not -- and all we can do without a cut-off is to say that some process may lead to more objectivity than some other. That is, whenever $\eta\neq 1$, one may not be able to say, safely,  that `objectivity has emerged'.
If we decide to choose an arbitrary bound for $\eta$ close enough to $1$, above which we are confident to say that `objectivity has emerged', we also face some issues.}}

First, the definition of objectivity itself becomes subjective, since what one considers `close enough to $1$' is certainly not objective. Second, and most importantly,  we want any cut-off signalling objectivity to ensure emergence of noncontextuality as well, as this provides a well-motivated and broad notion of classicality. Our results tell us that sufficiently high values of $\eta$ ensure this is true -- but how high $\eta$ should be depends on the monitored observable $\{\tilde{E}_k\}_k$. In fact, for any arbitrary fixed bound $\eta$ one chooses, there is always a central system and specific dynamics in which $\eta<\hat{P}[(\tilde{E}_k)_k]$ (see Appendix \ref{App:DecAndSSB} for an example of this behaviour on complete decoherence dynamics). In these cases, we have no guarantee that noncontextuality emerges, even though $\eta$ would be `close enough to $1$' according to an arbitrary choice. 

To avoid those problems, and in view of our results, we propose $\hat{P}$ to be such a cut-off for classical objectivity. This proposal gives a bound to talk about objectivity which is neither subjective nor signals objectivity without noncontextuality. 

\begin{definition}[Emergence of classical objectivity under QD$_\eta$] \label{def:cut-off}
Consider a process QD$_\eta$, with
observable $\{\tilde{E}_k\}_k$. We say that \textbf{classical objectivity} has emerged under this process if $\eta>~\hat{P}[(\tilde{E}_k)_k]$.
\end{definition}

As said above, setting a {non-subjective} bound for objectivity is {crucial}: we are now able to say whether classical objectivity has emerged or not in QD$_\eta$ processes depending on $\eta$. By taking $\hat{P}$ as such cut-off, we are sure noncontextuality has emerged and the bound \emph{depends on the monitored observable}, $\{\tilde{E}_k\}_k$. This makes sense, as insensitivity to the specifics of the interaction is rather artificial.
In other words, central systems with different dimensions and different dynamics, leading to selection of different pointer observables, \emph{should} impact the decision on whether classical objectivity has arisen or not. This is also the case in other approaches to quantum Darwinism, where mutual information between the system and a fragment of the environment is compared to the von Neumann entropy of the system, which depends on the initial state of $A$ and on the dynamics \cite{Zurek_2003_Review_QD0}.

Moreover, since $\eta>\hat{P}$ is not a necessary condition for noncontextuality to have emerged, but a sufficient condition obeyed with sufficiently high $p_{\rm guess}$ for each Bob, it sets a good threshold that will not be artificially close to zero -- this could happen if we only asked for emergence of noncontextuality, as affinely independent states could be almost indistinguishable. Indeed, the perfect decoherence example shows that $\hat{P}$ can be fairly close to $1$, setting a high lower bound on the probability that all Bobs will see the same outcome when making the correct measurements. Therefore, the above definition proposes a new and physically well-motivated way to treat emergence of Darwinist objectivity within the framework of \citet*{BPH_2015_QD}. 
}
\subsection{The case of State Spectrum Broadcasting}
{
Before we conclude, let us mention some interesting byproduct of Theorem~\ref{theorem: NCunderQD}:  noncontextuality necessarily emerges under the process of State Spectrum Broadcasting \cite{Horodecki_SSB_2015}. This process was proposed as an alternative way to reach objectivity from the interaction of quantum systems. In such a process, the central system suffers full decoherence and its spectrum is perfectly broadcast to the environment.  As we show in the Appendix (see \ref{App:DecAndSSB} for more details), it can be described as a special form of QD$_\eta$ process adapted to deal with a post-interaction state of the central system $A$. It demands full distinguishability, so we have $\eta=1$ and full decoherence implies $\tilde{E}_k=\ket{k}\bra{k}$, with $\{\ket{k}\bra{k}\}_k$ a basis of $\mathcal{H}_A$. This form of $\tilde{E}_k$ leads to a $\hat{P}<1=\eta$. By applying Theorem~\ref{theorem: NCunderQD}, we get
\begin{corollary}[Emergence of noncontextuality under State Spectrum Broadcasting]
If the EW$_t$-dynamics leads to the occurrence of State Spectrum Broadcasting process for arbitrary initial states $\rho^A$, all the Bobs can construct a noncontextual ontological model to their statistics.
\begin{proof}
As we mentioned above, if a State Spectrum Broadcasting process occurs, we have $\eta=1$ and $\tilde{E}_k$ is a projection into a basis of $\mathcal{H}_A$. We prove in the appendix that, for such kind of observable $\tilde{E}_k$, the bound $\hat{P}$ gives $\hat{P}[(\ket{k}\bra{k})_k]=1-\frac{1}{2{\rm dim}(\mathcal{H}_A)}$. Since this is smaller than $\eta=1$, the corollary follows. See the Appendix~\ref{App:DecAndSSB} for more details.
\end{proof}
\end{corollary}
}

\section{Conclusion}

{This paper shows that a particular expression of non-classicality may emerge from quantum Darwinism processes. }{ More precisely, we have proved that non-contextuality is implied for every observer in $S_t$ in quantum Darwinism processes $QD_{\eta}$ with an appropriate distinguishability rate $\eta > \hat{P}$. The cut-off $\hat{P}$ depending only on the monitored POVM $\{\tilde{E}_k\}_k$}. 

This result is {crucially} important{, as} it shows that QD$_\eta$ processes with high $\eta$, besides leading to objectivity { (i,e, high probability of independent observers seeing the same outcome when implementing the adequate measurements)}, also  {implies that} each of them can  {adopt} noncontextual explanations {for their aggregated statistics}. {Moreover, they can construct noncontextual models  with}  {some} aspects {in common: the ontic space $\Lambda$ and distributions $\mu_{P'}$, as defined in Eq.~\eqref{eq: NCOMQD}, are independent of $B_j$.} Such noncontextual explanations may provide a classical view underpinning the probabilities each Bob sees, taking away the need for any nonclassical (in particular, quantum) {description} of the central system. {Interestingly, these conclusions also hold true for State Spectrum Broadcasting, an alternative process for emergence of objectivity in the quantum realm.} 

{ Another important byproduct of our results is the { fact that one can} use $\hat{P}$ as a cut-off for emergence of classical objectivity in QD$_\eta$ processes. This brings a new  perspective
to the Darwinist objectivity proposed by \citet*{BPH_2015_QD}, in which the specifics of the ongoing dynamics set a bound for $\eta$, above which objectivity and noncontextuality have emerged. Indeed, definition~\ref{def:cut-off} suggests that emergence of noncontextuality should also be considered as an important benchmark for emergent classicality within Darwinist processes.}

Remarkably, { even if Eq.\eqref{eq:QDeta} is only obeyed {for} small values of $\eta$, so objectivity under a Darwinist process essentially} does not occur, noncontextuality can still emerge, as long as $\{\sigma^{B_j}_k\}{_k}$ are affinely independent. Indeed, even almost indistinguishable states can be affinely independent; in such cases { it is fair to say that objectivity will not emerge}  whereas the probabilities can still be explained by noncontextual ontological models.
{In this sense}, our {findings} reinforce the idea that the { Environment as a Witness} paradigm alone already represents a step towards {emergent} classicality. One can also see this by using the minimum dimension of Bobs' systems as a figure of merit:  {on} one hand, for noncontextuality to emerge due to an EW$_t$-dynamics together with affine independence of $\{\sigma_k\}{_k}$, one must have ${\rm dim}(\mathcal{H}_{B_j})^2\geq k_{\rm max}$; {on} the other hand, our bound identifying `almost orthogonality' among $\{\sigma^{B_j}_k\}{_k}$ (which implies affine independence), implies ${\rm dim}(\mathcal{H}_{B_j})\geq k_{\rm max}$ (see remark after Lemma \ref{lemma: OperationalSignAI}). Therefore, a portion of the environment that cannot carry as much classical information about $A$ as required by Darwinism (to lead to `enough' objectivity of outcomes), can still lead to noncontextuality. 

Interesting perspectives also arise from this work. {A first question, similar to the problem of how often interactions lead to $\eta\approx1$, is if generic interactions lead to the condition $\eta>\hat{P}$.} Second, it would be interesting to understand if strong quantum Darwinism \cite{Le_StrongQD_2019} -- defined as a set of conditions on the mutual information and discord of system-environment states -- also lead to noncontextuality. Third, it would be interesting to consider what is the impact of some special interactions on this emergence of classicality, such as those leading to non-markovianity, whose impact on Darwinist processes is a matter of debate~\cite{,Galve__NMarkovQD_2016,Sheilla_QDNMarkov_2019}. { Fourth, analyzing how to leverage the results of this work to the generalized probabilistic theories framework would certainly be of interest, as well as possible connections with state discrimination in this framework. Indeed, this could provide a way to understand emergence of noncontextuality due to dynamics without resorting on the quantum formalism.}

{In sum, t}his work is strongly based on the description of the { Environment as a Witness} paradigm and QD process as given by \citet*{BPH_2015_QD} -- which shows the generic impacts of these dynamics. Therefore, our results also point to {the} generic emergence of noncontextuality under the same paradigm, which is {undeniably} a good approximation of our experiences as {classical} observers \cite{Ollivier_2005_EW,zurek_2007_RelativeQDReview}. Thus, we may conclude that our noncontextual experience as observers, naturally embodied in classical physics, may emerge from this indirect interaction with the systems surrounding us. This work also brings new reasons to deem as classical the features emerging under {{QD$_\eta$}} processes even if the POVM pointer $\{\tilde{E}_k\}_k$ is {informationally complete}. 

\begin{acknowledgments}
We thank Renan Cunha for important discussions about relations to minimum-error state discrimination. 
{All authors are grateful for the reviewers for important suggestions that helped improve presentation and deepen the discussions. In particular, for pointing out the possibility of using $\hat{P}$ as  a cut-off.} RDB, RW, BA and MTC are thankful to the Brazilian agencies CNPq and CAPES for financial support. This work is part of the Brazilian National Institute for Science and Technology of
Quantum Information (INCT-IQ). RW also acknowledges financial support by the Portuguese funding institution FCT - Funda\c{c}\~{a}o para Ciência e Tecnologia, via scholarship No. SFRH/BD/151199/2021. RDB acknowledges funding from grant $2016/24162-8$, S\~{a}o Paulo Research Foundation (FAPESP). This research was supported by Grant No. FQXi-RFP-IPW-1905 from the Foundational Questions Institute and Fetzer Franklin Fund, a donor advised fund of Silicon Valley Community Foundation. This work was also supported by the National Research, Development and Innovation Office of Hungary (NKFIH) through the Quantum Information National Laboratory of Hungary and through Grant No. FK 135220.
\end{acknowledgments}

\bibliographystyle{apsrev4-1}
\bibliography{Bib}

\begin{thebibliography}{58}%
\makeatletter
\providecommand \@ifxundefined [1]{%
 \@ifx{#1\undefined}
}%
\providecommand \@ifnum [1]{%
 \ifnum #1\expandafter \@firstoftwo
 \else \expandafter \@secondoftwo
 \fi
}%
\providecommand \@ifx [1]{%
 \ifx #1\expandafter \@firstoftwo
 \else \expandafter \@secondoftwo
 \fi
}%
\providecommand \natexlab [1]{#1}%
\providecommand \enquote  [1]{``#1''}%
\providecommand \bibnamefont  [1]{#1}%
\providecommand \bibfnamefont [1]{#1}%
\providecommand \citenamefont [1]{#1}%
\providecommand \href@noop [0]{\@secondoftwo}%
\providecommand \href [0]{\begingroup \@sanitize@url \@href}%
\providecommand \@href[1]{\@@startlink{#1}\@@href}%
\providecommand \@@href[1]{\endgroup#1\@@endlink}%
\providecommand \@sanitize@url [0]{\catcode `\\12\catcode `\$12\catcode
  `\&12\catcode `\#12\catcode `\^12\catcode `\_12\catcode `\%12\relax}%
\providecommand \@@startlink[1]{}%
\providecommand \@@endlink[0]{}%
\providecommand \url  [0]{\begingroup\@sanitize@url \@url }%
\providecommand \@url [1]{\endgroup\@href {#1}{\urlprefix }}%
\providecommand \urlprefix  [0]{URL }%
\providecommand \Eprint [0]{\href }%
\providecommand \doibase [0]{http://dx.doi.org/}%
\providecommand \selectlanguage [0]{\@gobble}%
\providecommand \bibinfo  [0]{\@secondoftwo}%
\providecommand \bibfield  [0]{\@secondoftwo}%
\providecommand \translation [1]{[#1]}%
\providecommand \BibitemOpen [0]{}%
\providecommand \bibitemStop [0]{}%
\providecommand \bibitemNoStop [0]{.\EOS\space}%
\providecommand \EOS [0]{\spacefactor3000\relax}%
\providecommand \BibitemShut  [1]{\csname bibitem#1\endcsname}%
\let\auto@bib@innerbib\@empty
\bibitem [{\citenamefont {Bell}(1966)}]{Bell66}%
  \BibitemOpen
  \bibfield  {author} {\bibinfo {author} {\bibfnamefont {J.~S.}\ \bibnamefont
  {Bell}},\ }\href {\doibase 10.1103/RevModPhys.38.447} {\bibfield  {journal}
  {\bibinfo  {journal} {Rev. Mod. Phys.}\ }\textbf {\bibinfo {volume} {38}},\
  \bibinfo {pages} {447} (\bibinfo {year} {1966})}\BibitemShut {NoStop}%
\bibitem [{\citenamefont {Bell}(1964)}]{Bell64}%
  \BibitemOpen
  \bibfield  {author} {\bibinfo {author} {\bibfnamefont {J.~S.}\ \bibnamefont
  {Bell}},\ }\href {\doibase 10.1103/PhysicsPhysiqueFizika.1.195} {\bibfield
  {journal} {\bibinfo  {journal} {Physics Physique Fizika}\ }\textbf {\bibinfo
  {volume} {1}},\ \bibinfo {pages} {195} (\bibinfo {year} {1964})}\BibitemShut
  {NoStop}%
\bibitem [{\citenamefont {Jennings}\ and\ \citenamefont {Leifer}(2016)}]{JL16}%
  \BibitemOpen
  \bibfield  {author} {\bibinfo {author} {\bibfnamefont {D.}~\bibnamefont
  {Jennings}}\ and\ \bibinfo {author} {\bibfnamefont {M.}~\bibnamefont
  {Leifer}},\ }\href {\doibase 10.1080/00107514.2015.1063233} {\bibfield
  {journal} {\bibinfo  {journal} {Contemporary Physics}\ }\textbf {\bibinfo
  {volume} {57}},\ \bibinfo {pages} {60} (\bibinfo {year} {2016})},\ \Eprint
  {http://arxiv.org/abs/https://doi.org/10.1080/00107514.2015.1063233}
  {https://doi.org/10.1080/00107514.2015.1063233} \BibitemShut {NoStop}%
\bibitem [{\citenamefont {Goh}\ \emph {et~al.}(2018)\citenamefont {Goh},
  \citenamefont {Kaniewski}, \citenamefont {Wolfe}, \citenamefont {V\'ertesi},
  \citenamefont {Wu}, \citenamefont {Cai}, \citenamefont {Liang},\ and\
  \citenamefont {Scarani}}]{GohEtAl18}%
  \BibitemOpen
  \bibfield  {author} {\bibinfo {author} {\bibfnamefont {K.~T.}\ \bibnamefont
  {Goh}}, \bibinfo {author} {\bibfnamefont {J.~m.~k.}\ \bibnamefont
  {Kaniewski}}, \bibinfo {author} {\bibfnamefont {E.}~\bibnamefont {Wolfe}},
  \bibinfo {author} {\bibfnamefont {T.}~\bibnamefont {V\'ertesi}}, \bibinfo
  {author} {\bibfnamefont {X.}~\bibnamefont {Wu}}, \bibinfo {author}
  {\bibfnamefont {Y.}~\bibnamefont {Cai}}, \bibinfo {author} {\bibfnamefont
  {Y.-C.}\ \bibnamefont {Liang}}, \ and\ \bibinfo {author} {\bibfnamefont
  {V.}~\bibnamefont {Scarani}},\ }\href {\doibase 10.1103/PhysRevA.97.022104}
  {\bibfield  {journal} {\bibinfo  {journal} {Phys. Rev. A}\ }\textbf {\bibinfo
  {volume} {97}},\ \bibinfo {pages} {022104} (\bibinfo {year}
  {2018})}\BibitemShut {NoStop}%
\bibitem [{\citenamefont {Brunner}\ \emph {et~al.}(2014)\citenamefont
  {Brunner}, \citenamefont {Cavalcanti}, \citenamefont {Pironio}, \citenamefont
  {Scarani},\ and\ \citenamefont {Wehner}}]{BCPSW14}%
  \BibitemOpen
  \bibfield  {author} {\bibinfo {author} {\bibfnamefont {N.}~\bibnamefont
  {Brunner}}, \bibinfo {author} {\bibfnamefont {D.}~\bibnamefont {Cavalcanti}},
  \bibinfo {author} {\bibfnamefont {S.}~\bibnamefont {Pironio}}, \bibinfo
  {author} {\bibfnamefont {V.}~\bibnamefont {Scarani}}, \ and\ \bibinfo
  {author} {\bibfnamefont {S.}~\bibnamefont {Wehner}},\ }\href {\doibase
  10.1103/RevModPhys.86.419} {\bibfield  {journal} {\bibinfo  {journal} {Rev.
  Mod. Phys.}\ }\textbf {\bibinfo {volume} {86}},\ \bibinfo {pages} {419}
  (\bibinfo {year} {2014})}\BibitemShut {NoStop}%
\bibitem [{\citenamefont {Pusey}\ \emph {et~al.}(2012)\citenamefont {Pusey},
  \citenamefont {Barrett},\ and\ \citenamefont {Rudolph}}]{PBR12}%
  \BibitemOpen
  \bibfield  {author} {\bibinfo {author} {\bibfnamefont {M.~F.}\ \bibnamefont
  {Pusey}}, \bibinfo {author} {\bibfnamefont {J.}~\bibnamefont {Barrett}}, \
  and\ \bibinfo {author} {\bibfnamefont {T.}~\bibnamefont {Rudolph}},\ }\href
  {\doibase 10.1038/nphys2309} {\bibfield  {journal} {\bibinfo  {journal}
  {Nature Physics}\ }\textbf {\bibinfo {volume} {8}},\ \bibinfo {pages}
  {475–478} (\bibinfo {year} {2012})}\BibitemShut {NoStop}%
\bibitem [{\citenamefont {Leifer}(2014)}]{Leifer14}%
  \BibitemOpen
  \bibfield  {author} {\bibinfo {author} {\bibfnamefont {M.~S.}\ \bibnamefont
  {Leifer}},\ }\href {\doibase 10.12743/quanta.v3i1.22} {\bibfield  {journal}
  {\bibinfo  {journal} {Quanta}\ }\textbf {\bibinfo {volume} {3}},\ \bibinfo
  {pages} {67} (\bibinfo {year} {2014})}\BibitemShut {NoStop}%
\bibitem [{\citenamefont {Spekkens}(2019)}]{spekkens2019_Leibniz}%
  \BibitemOpen
  \bibfield  {author} {\bibinfo {author} {\bibfnamefont {R.~W.}\ \bibnamefont
  {Spekkens}},\ }\href@noop {} {\enquote {\bibinfo {title} {The ontological
  identity of empirical indiscernibles: Leibniz's methodological principle and
  its significance in the work of {E}instein},}\ } (\bibinfo {year} {2019}),\
  \Eprint {http://arxiv.org/abs/1909.04628} {arXiv:1909.04628
  [physics.hist-ph]} \BibitemShut {NoStop}%
\bibitem [{\citenamefont {Spekkens}(2008)}]{Spekkens_2008_C&Negativity}%
  \BibitemOpen
  \bibfield  {author} {\bibinfo {author} {\bibfnamefont {R.~W.}\ \bibnamefont
  {Spekkens}},\ }\href {\doibase 10.1103/physrevlett.101.020401} {\bibfield
  {journal} {\bibinfo  {journal} {Physical Review Letters}\ }\textbf {\bibinfo
  {volume} {101}} (\bibinfo {year} {2008}),\
  10.1103/physrevlett.101.020401}\BibitemShut {NoStop}%
\bibitem [{\citenamefont {Spekkens}(2005)}]{Spekkens_2005_C}%
  \BibitemOpen
  \bibfield  {author} {\bibinfo {author} {\bibfnamefont {R.~W.}\ \bibnamefont
  {Spekkens}},\ }\href {\doibase 10.1103/physreva.71.052108} {\bibfield
  {journal} {\bibinfo  {journal} {Physical Review A}\ }\textbf {\bibinfo
  {volume} {71}} (\bibinfo {year} {2005}),\
  10.1103/physreva.71.052108}\BibitemShut {NoStop}%
\bibitem [{\citenamefont {Schmid}\ \emph {et~al.}(2020)\citenamefont {Schmid},
  \citenamefont {Selby}, \citenamefont {Pusey},\ and\ \citenamefont
  {Spekkens}}]{schmid2020_StructureTheorem}%
  \BibitemOpen
  \bibfield  {author} {\bibinfo {author} {\bibfnamefont {D.}~\bibnamefont
  {Schmid}}, \bibinfo {author} {\bibfnamefont {J.~H.}\ \bibnamefont {Selby}},
  \bibinfo {author} {\bibfnamefont {M.~F.}\ \bibnamefont {Pusey}}, \ and\
  \bibinfo {author} {\bibfnamefont {R.~W.}\ \bibnamefont {Spekkens}},\
  }\href@noop {} {\enquote {\bibinfo {title} {A structure theorem for
  generalized-noncontextual ontological models},}\ } (\bibinfo {year} {2020}),\
  \Eprint {http://arxiv.org/abs/2005.07161} {arXiv:2005.07161 [quant-ph]}
  \BibitemShut {NoStop}%
\bibitem [{\citenamefont {Zurek}(2003)}]{Zurek_2003_Review_QD0}%
  \BibitemOpen
  \bibfield  {author} {\bibinfo {author} {\bibfnamefont {W.~H.}\ \bibnamefont
  {Zurek}},\ }\href {\doibase 10.1103/revmodphys.75.715} {\bibfield  {journal}
  {\bibinfo  {journal} {Reviews of Modern Physics}\ }\textbf {\bibinfo {volume}
  {75}},\ \bibinfo {pages} {715–775} (\bibinfo {year} {2003})}\BibitemShut
  {NoStop}%
\bibitem [{\citenamefont {Zurek}(1981)}]{Zurek_1981_PointerStates}%
  \BibitemOpen
  \bibfield  {author} {\bibinfo {author} {\bibfnamefont {W.~H.}\ \bibnamefont
  {Zurek}},\ }\href {\doibase 10.1103/PhysRevD.24.1516} {\bibfield  {journal}
  {\bibinfo  {journal} {Phys. Rev. D}\ }\textbf {\bibinfo {volume} {24}},\
  \bibinfo {pages} {1516} (\bibinfo {year} {1981})}\BibitemShut {NoStop}%
\bibitem [{\citenamefont {Zurek}(2007)}]{zurek_2007_RelativeQDReview}%
  \BibitemOpen
  \bibfield  {author} {\bibinfo {author} {\bibfnamefont {W.~H.}\ \bibnamefont
  {Zurek}},\ }\href@noop {} {\enquote {\bibinfo {title} {Relative states and
  the environment: Einselection, envariance, quantum darwinism, and the
  existential interpretation},}\ } (\bibinfo {year} {2007}),\ \Eprint
  {http://arxiv.org/abs/0707.2832} {arXiv:0707.2832 [quant-ph]} \BibitemShut
  {NoStop}%
\bibitem [{\citenamefont {Brandão}\ \emph {et~al.}(2015)\citenamefont
  {Brandão}, \citenamefont {Piani},\ and\ \citenamefont
  {Horodecki}}]{BPH_2015_QD}%
  \BibitemOpen
  \bibfield  {author} {\bibinfo {author} {\bibfnamefont {F.~G. S.~L.}\
  \bibnamefont {Brandão}}, \bibinfo {author} {\bibfnamefont {M.}~\bibnamefont
  {Piani}}, \ and\ \bibinfo {author} {\bibfnamefont {P.}~\bibnamefont
  {Horodecki}},\ }\href {\doibase 10.1038/ncomms8908} {\bibfield  {journal}
  {\bibinfo  {journal} {Nature Communications}\ }\textbf {\bibinfo {volume}
  {6}} (\bibinfo {year} {2015}),\ 10.1038/ncomms8908}\BibitemShut {NoStop}%
\bibitem [{\citenamefont {Horodecki}\ \emph {et~al.}(2015)\citenamefont
  {Horodecki}, \citenamefont {Korbicz},\ and\ \citenamefont
  {Horodecki}}]{Horodecki_SSB_2015}%
  \BibitemOpen
  \bibfield  {author} {\bibinfo {author} {\bibfnamefont {R.}~\bibnamefont
  {Horodecki}}, \bibinfo {author} {\bibfnamefont {J.~K.}\ \bibnamefont
  {Korbicz}}, \ and\ \bibinfo {author} {\bibfnamefont {P.}~\bibnamefont
  {Horodecki}},\ }\href {\doibase 10.1103/PhysRevA.91.032122} {\bibfield
  {journal} {\bibinfo  {journal} {Phys. Rev. A}\ }\textbf {\bibinfo {volume}
  {91}},\ \bibinfo {pages} {032122} (\bibinfo {year} {2015})}\BibitemShut
  {NoStop}%
\bibitem [{\citenamefont {Le}\ and\ \citenamefont
  {Olaya-Castro}(2019)}]{Le_StrongQD_2019}%
  \BibitemOpen
  \bibfield  {author} {\bibinfo {author} {\bibfnamefont {T.~P.}\ \bibnamefont
  {Le}}\ and\ \bibinfo {author} {\bibfnamefont {A.}~\bibnamefont
  {Olaya-Castro}},\ }\href {\doibase 10.1103/PhysRevLett.122.010403} {\bibfield
   {journal} {\bibinfo  {journal} {Phys. Rev. Lett.}\ }\textbf {\bibinfo
  {volume} {122}},\ \bibinfo {pages} {010403} (\bibinfo {year}
  {2019})}\BibitemShut {NoStop}%
\bibitem [{\citenamefont {Contreras-Tejada}\ \emph {et~al.}(2021)\citenamefont
  {Contreras-Tejada}, \citenamefont {Scarpa}, \citenamefont {Kubicki},
  \citenamefont {Brandenburger},\ and\ \citenamefont
  {Mura}}]{ContrerasTejada2021_Agreement}%
  \BibitemOpen
  \bibfield  {author} {\bibinfo {author} {\bibfnamefont {P.}~\bibnamefont
  {Contreras-Tejada}}, \bibinfo {author} {\bibfnamefont {G.}~\bibnamefont
  {Scarpa}}, \bibinfo {author} {\bibfnamefont {A.~M.}\ \bibnamefont {Kubicki}},
  \bibinfo {author} {\bibfnamefont {A.}~\bibnamefont {Brandenburger}}, \ and\
  \bibinfo {author} {\bibfnamefont {P.~L.}\ \bibnamefont {Mura}},\ }\href@noop
  {} {\enquote {\bibinfo {title} {Agreement between observers: a physical
  principle?}}\ } (\bibinfo {year} {2021}),\ \Eprint
  {http://arxiv.org/abs/2102.08966} {arXiv:2102.08966 [quant-ph]} \BibitemShut
  {NoStop}%
\bibitem [{\citenamefont {Mironowicz}\ \emph {et~al.}(2017)\citenamefont
  {Mironowicz}, \citenamefont {Korbicz},\ and\ \citenamefont
  {Horodecki}}]{Mironowicz2017_InterSubj}%
  \BibitemOpen
  \bibfield  {author} {\bibinfo {author} {\bibfnamefont {P.}~\bibnamefont
  {Mironowicz}}, \bibinfo {author} {\bibfnamefont {J.~K.}\ \bibnamefont
  {Korbicz}}, \ and\ \bibinfo {author} {\bibfnamefont {P.}~\bibnamefont
  {Horodecki}},\ }\href {\doibase 10.1103/PhysRevLett.118.150501} {\bibfield
  {journal} {\bibinfo  {journal} {Phys. Rev. Lett.}\ }\textbf {\bibinfo
  {volume} {118}},\ \bibinfo {pages} {150501} (\bibinfo {year}
  {2017})}\BibitemShut {NoStop}%
\bibitem [{\citenamefont {Ollivier}\ \emph {et~al.}(2005)\citenamefont
  {Ollivier}, \citenamefont {Poulin},\ and\ \citenamefont
  {Zurek}}]{Ollivier_2005_EW}%
  \BibitemOpen
  \bibfield  {author} {\bibinfo {author} {\bibfnamefont {H.}~\bibnamefont
  {Ollivier}}, \bibinfo {author} {\bibfnamefont {D.}~\bibnamefont {Poulin}}, \
  and\ \bibinfo {author} {\bibfnamefont {W.~H.}\ \bibnamefont {Zurek}},\ }\href
  {\doibase 10.1103/physreva.72.042113} {\bibfield  {journal} {\bibinfo
  {journal} {Physical Review A}\ }\textbf {\bibinfo {volume} {72}} (\bibinfo
  {year} {2005}),\ 10.1103/physreva.72.042113}\BibitemShut {NoStop}%
\bibitem [{\citenamefont {Bény}(2008)}]{Beny_2008_unsharp}%
  \BibitemOpen
  \bibfield  {author} {\bibinfo {author} {\bibfnamefont {C.}~\bibnamefont
  {Bény}},\ }\href@noop {} {\enquote {\bibinfo {title} {Unsharp pointer
  observables and the structure of decoherence},}\ } (\bibinfo {year} {2008}),\
  \Eprint {http://arxiv.org/abs/0802.0685} {arXiv:0802.0685 [quant-ph]}
  \BibitemShut {NoStop}%
\bibitem [{\citenamefont {Zwolak}\ and\ \citenamefont
  {Zurek}(2017)}]{Zwolak_2017_IrrelevanceofIrrelevant}%
  \BibitemOpen
  \bibfield  {author} {\bibinfo {author} {\bibfnamefont {M.}~\bibnamefont
  {Zwolak}}\ and\ \bibinfo {author} {\bibfnamefont {W.~H.}\ \bibnamefont
  {Zurek}},\ }\href {\doibase 10.1103/PhysRevA.95.030101} {\bibfield  {journal}
  {\bibinfo  {journal} {Phys. Rev. A}\ }\textbf {\bibinfo {volume} {95}},\
  \bibinfo {pages} {030101} (\bibinfo {year} {2017})}\BibitemShut {NoStop}%
\bibitem [{\citenamefont {Mandel}\ and\ \citenamefont
  {Wolf}(1995)}]{mandel_wolf_1995}%
  \BibitemOpen
  \bibfield  {author} {\bibinfo {author} {\bibfnamefont {L.}~\bibnamefont
  {Mandel}}\ and\ \bibinfo {author} {\bibfnamefont {E.}~\bibnamefont {Wolf}},\
  }\href {\doibase 10.1017/CBO9781139644105} {\emph {\bibinfo {title} {Optical
  Coherence and Quantum Optics}}}\ (\bibinfo  {publisher} {Cambridge University
  Press},\ \bibinfo {year} {1995})\BibitemShut {NoStop}%
\bibitem [{\citenamefont {Wigner}(1932)}]{Wigner_Quasiprobability}%
  \BibitemOpen
  \bibfield  {author} {\bibinfo {author} {\bibfnamefont {E.}~\bibnamefont
  {Wigner}},\ }\href {\doibase 10.1103/PhysRev.40.749} {\bibfield  {journal}
  {\bibinfo  {journal} {Phys. Rev.}\ }\textbf {\bibinfo {volume} {40}},\
  \bibinfo {pages} {749} (\bibinfo {year} {1932})}\BibitemShut {NoStop}%
\bibitem [{\citenamefont {Barrett}(2007)}]{Barrett_GPT}%
  \BibitemOpen
  \bibfield  {author} {\bibinfo {author} {\bibfnamefont {J.}~\bibnamefont
  {Barrett}},\ }\href {\doibase 10.1103/PhysRevA.75.032304} {\bibfield
  {journal} {\bibinfo  {journal} {Phys. Rev. A}\ }\textbf {\bibinfo {volume}
  {75}},\ \bibinfo {pages} {032304} (\bibinfo {year} {2007})}\BibitemShut
  {NoStop}%
\bibitem [{\citenamefont {Janotta}\ and\ \citenamefont
  {Hinrichsen}(2014)}]{Janotta_GPT}%
  \BibitemOpen
  \bibfield  {author} {\bibinfo {author} {\bibfnamefont {P.}~\bibnamefont
  {Janotta}}\ and\ \bibinfo {author} {\bibfnamefont {H.}~\bibnamefont
  {Hinrichsen}},\ }\href {\doibase 10.1088/1751-8113/47/32/323001} {\bibfield
  {journal} {\bibinfo  {journal} {Journal of Physics A: Mathematical and
  Theoretical}\ }\textbf {\bibinfo {volume} {47}},\ \bibinfo {pages} {323001}
  (\bibinfo {year} {2014})}\BibitemShut {NoStop}%
\bibitem [{\citenamefont {Schmid}\ \emph
  {et~al.}(2021{\natexlab{a}})\citenamefont {Schmid}, \citenamefont {Selby},
  \citenamefont {Wolfe}, \citenamefont {Kunjwal},\ and\ \citenamefont
  {Spekkens}}]{schmid2021characterization}%
  \BibitemOpen
  \bibfield  {author} {\bibinfo {author} {\bibfnamefont {D.}~\bibnamefont
  {Schmid}}, \bibinfo {author} {\bibfnamefont {J.~H.}\ \bibnamefont {Selby}},
  \bibinfo {author} {\bibfnamefont {E.}~\bibnamefont {Wolfe}}, \bibinfo
  {author} {\bibfnamefont {R.}~\bibnamefont {Kunjwal}}, \ and\ \bibinfo
  {author} {\bibfnamefont {R.~W.}\ \bibnamefont {Spekkens}},\ }\href@noop {}
  {\bibfield  {journal} {\bibinfo  {journal} {PRX Quantum}\ }\textbf {\bibinfo
  {volume} {2}},\ \bibinfo {pages} {010331} (\bibinfo {year}
  {2021}{\natexlab{a}})}\BibitemShut {NoStop}%
\bibitem [{\citenamefont {Bartlett}\ \emph {et~al.}(2012)\citenamefont
  {Bartlett}, \citenamefont {Rudolph},\ and\ \citenamefont
  {Spekkens}}]{bartlett2012reconstruction}%
  \BibitemOpen
  \bibfield  {author} {\bibinfo {author} {\bibfnamefont {S.~D.}\ \bibnamefont
  {Bartlett}}, \bibinfo {author} {\bibfnamefont {T.}~\bibnamefont {Rudolph}}, \
  and\ \bibinfo {author} {\bibfnamefont {R.~W.}\ \bibnamefont {Spekkens}},\
  }\href@noop {} {\bibfield  {journal} {\bibinfo  {journal} {Physical Review
  A}\ }\textbf {\bibinfo {volume} {86}},\ \bibinfo {pages} {012103} (\bibinfo
  {year} {2012})}\BibitemShut {NoStop}%
\bibitem [{\citenamefont {Kunjwal}\ and\ \citenamefont
  {Spekkens}(2015)}]{KunjwalSpekkens_KSwithoutDeterminism}%
  \BibitemOpen
  \bibfield  {author} {\bibinfo {author} {\bibfnamefont {R.}~\bibnamefont
  {Kunjwal}}\ and\ \bibinfo {author} {\bibfnamefont {R.~W.}\ \bibnamefont
  {Spekkens}},\ }\href {\doibase 10.1103/PhysRevLett.115.110403} {\bibfield
  {journal} {\bibinfo  {journal} {Phys. Rev. Lett.}\ }\textbf {\bibinfo
  {volume} {115}},\ \bibinfo {pages} {110403} (\bibinfo {year}
  {2015})}\BibitemShut {NoStop}%
\bibitem [{\citenamefont {Kunjwal}\ and\ \citenamefont
  {Spekkens}(2018)}]{kunjwal2018statistical}%
  \BibitemOpen
  \bibfield  {author} {\bibinfo {author} {\bibfnamefont {R.}~\bibnamefont
  {Kunjwal}}\ and\ \bibinfo {author} {\bibfnamefont {R.~W.}\ \bibnamefont
  {Spekkens}},\ }\href@noop {} {\bibfield  {journal} {\bibinfo  {journal}
  {Physical Review A}\ }\textbf {\bibinfo {volume} {97}},\ \bibinfo {pages}
  {052110} (\bibinfo {year} {2018})}\BibitemShut {NoStop}%
\bibitem [{\citenamefont {Kunjwal}(2019)}]{Kunjwal_2019}%
  \BibitemOpen
  \bibfield  {author} {\bibinfo {author} {\bibfnamefont {R.}~\bibnamefont
  {Kunjwal}},\ }\href {\doibase 10.22331/q-2019-09-09-184} {\bibfield
  {journal} {\bibinfo  {journal} {Quantum}\ }\textbf {\bibinfo {volume} {3}},\
  \bibinfo {pages} {184} (\bibinfo {year} {2019})}\BibitemShut {NoStop}%
\bibitem [{\citenamefont {Kunjwal}(2020)}]{Kunjwal_2020}%
  \BibitemOpen
  \bibfield  {author} {\bibinfo {author} {\bibfnamefont {R.}~\bibnamefont
  {Kunjwal}},\ }\href {\doibase 10.22331/q-2020-01-10-219} {\bibfield
  {journal} {\bibinfo  {journal} {Quantum}\ }\textbf {\bibinfo {volume} {4}},\
  \bibinfo {pages} {219} (\bibinfo {year} {2020})}\BibitemShut {NoStop}%
\bibitem [{\citenamefont {Schmid}\ \emph {et~al.}(2018)\citenamefont {Schmid},
  \citenamefont {Spekkens},\ and\ \citenamefont {Wolfe}}]{schmid2018all}%
  \BibitemOpen
  \bibfield  {author} {\bibinfo {author} {\bibfnamefont {D.}~\bibnamefont
  {Schmid}}, \bibinfo {author} {\bibfnamefont {R.~W.}\ \bibnamefont
  {Spekkens}}, \ and\ \bibinfo {author} {\bibfnamefont {E.}~\bibnamefont
  {Wolfe}},\ }\href {\doibase 10.1103/PhysRevA.97.062103} {\bibfield  {journal}
  {\bibinfo  {journal} {Phys. Rev. A}\ }\textbf {\bibinfo {volume} {97}},\
  \bibinfo {pages} {062103} (\bibinfo {year} {2018})}\BibitemShut {NoStop}%
\bibitem [{\citenamefont {Kochen}\ and\ \citenamefont {Specker}(1967)}]{KS}%
  \BibitemOpen
  \bibfield  {author} {\bibinfo {author} {\bibfnamefont {S.}~\bibnamefont
  {Kochen}}\ and\ \bibinfo {author} {\bibfnamefont {E.~P.}\ \bibnamefont
  {Specker}},\ }\href {http://www.jstor.org/stable/24902153} {\bibfield
  {journal} {\bibinfo  {journal} {Journal of Mathematics and Mechanics}\
  }\textbf {\bibinfo {volume} {17}},\ \bibinfo {pages} {59} (\bibinfo {year}
  {1967})}\BibitemShut {NoStop}%
\bibitem [{\citenamefont {Spekkens}\ \emph {et~al.}(2009)\citenamefont
  {Spekkens}, \citenamefont {Buzacott}, \citenamefont {Keehn}, \citenamefont
  {Toner},\ and\ \citenamefont {Pryde}}]{spekkens2009preparation}%
  \BibitemOpen
  \bibfield  {author} {\bibinfo {author} {\bibfnamefont {R.~W.}\ \bibnamefont
  {Spekkens}}, \bibinfo {author} {\bibfnamefont {D.~H.}\ \bibnamefont
  {Buzacott}}, \bibinfo {author} {\bibfnamefont {A.~J.}\ \bibnamefont {Keehn}},
  \bibinfo {author} {\bibfnamefont {B.}~\bibnamefont {Toner}}, \ and\ \bibinfo
  {author} {\bibfnamefont {G.~J.}\ \bibnamefont {Pryde}},\ }\href@noop {}
  {\bibfield  {journal} {\bibinfo  {journal} {Physical review letters}\
  }\textbf {\bibinfo {volume} {102}},\ \bibinfo {pages} {010401} (\bibinfo
  {year} {2009})}\BibitemShut {NoStop}%
\bibitem [{\citenamefont {Saha}\ and\ \citenamefont
  {Chaturvedi}(2019)}]{saha2019preparation}%
  \BibitemOpen
  \bibfield  {author} {\bibinfo {author} {\bibfnamefont {D.}~\bibnamefont
  {Saha}}\ and\ \bibinfo {author} {\bibfnamefont {A.}~\bibnamefont
  {Chaturvedi}},\ }\href@noop {} {\bibfield  {journal} {\bibinfo  {journal}
  {Physical Review A}\ }\textbf {\bibinfo {volume} {100}},\ \bibinfo {pages}
  {022108} (\bibinfo {year} {2019})}\BibitemShut {NoStop}%
\bibitem [{\citenamefont {Ghorai}\ and\ \citenamefont
  {Pan}(2018)}]{ghorai2018optimal}%
  \BibitemOpen
  \bibfield  {author} {\bibinfo {author} {\bibfnamefont {S.}~\bibnamefont
  {Ghorai}}\ and\ \bibinfo {author} {\bibfnamefont {A.}~\bibnamefont {Pan}},\
  }\href@noop {} {\bibfield  {journal} {\bibinfo  {journal} {Physical Review
  A}\ }\textbf {\bibinfo {volume} {98}},\ \bibinfo {pages} {032110} (\bibinfo
  {year} {2018})}\BibitemShut {NoStop}%
\bibitem [{\citenamefont {Ambainis}\ \emph {et~al.}(2019)\citenamefont
  {Ambainis}, \citenamefont {Banik}, \citenamefont {Chaturvedi}, \citenamefont
  {Kravchenko},\ and\ \citenamefont {Rai}}]{ambainis2019parity}%
  \BibitemOpen
  \bibfield  {author} {\bibinfo {author} {\bibfnamefont {A.}~\bibnamefont
  {Ambainis}}, \bibinfo {author} {\bibfnamefont {M.}~\bibnamefont {Banik}},
  \bibinfo {author} {\bibfnamefont {A.}~\bibnamefont {Chaturvedi}}, \bibinfo
  {author} {\bibfnamefont {D.}~\bibnamefont {Kravchenko}}, \ and\ \bibinfo
  {author} {\bibfnamefont {A.}~\bibnamefont {Rai}},\ }\href@noop {} {\bibfield
  {journal} {\bibinfo  {journal} {Quantum Information Processing}\ }\textbf
  {\bibinfo {volume} {18}},\ \bibinfo {pages} {1} (\bibinfo {year}
  {2019})}\BibitemShut {NoStop}%
\bibitem [{\citenamefont {Saha}\ \emph {et~al.}(2019)\citenamefont {Saha},
  \citenamefont {Horodecki},\ and\ \citenamefont
  {Paw{\l}owski}}]{saha2019state}%
  \BibitemOpen
  \bibfield  {author} {\bibinfo {author} {\bibfnamefont {D.}~\bibnamefont
  {Saha}}, \bibinfo {author} {\bibfnamefont {P.}~\bibnamefont {Horodecki}}, \
  and\ \bibinfo {author} {\bibfnamefont {M.}~\bibnamefont {Paw{\l}owski}},\
  }\href@noop {} {\bibfield  {journal} {\bibinfo  {journal} {New Journal of
  Physics}\ }\textbf {\bibinfo {volume} {21}},\ \bibinfo {pages} {093057}
  (\bibinfo {year} {2019})}\BibitemShut {NoStop}%
\bibitem [{\citenamefont {Chailloux}\ \emph {et~al.}(2016)\citenamefont
  {Chailloux}, \citenamefont {Kerenidis}, \citenamefont {Kundu},\ and\
  \citenamefont {Sikora}}]{chailloux2016optimal}%
  \BibitemOpen
  \bibfield  {author} {\bibinfo {author} {\bibfnamefont {A.}~\bibnamefont
  {Chailloux}}, \bibinfo {author} {\bibfnamefont {I.}~\bibnamefont
  {Kerenidis}}, \bibinfo {author} {\bibfnamefont {S.}~\bibnamefont {Kundu}}, \
  and\ \bibinfo {author} {\bibfnamefont {J.}~\bibnamefont {Sikora}},\
  }\href@noop {} {\bibfield  {journal} {\bibinfo  {journal} {New Journal of
  Physics}\ }\textbf {\bibinfo {volume} {18}},\ \bibinfo {pages} {045003}
  (\bibinfo {year} {2016})}\BibitemShut {NoStop}%
\bibitem [{\citenamefont {Schmid}\ and\ \citenamefont
  {Spekkens}(2018)}]{schmid2018contextual}%
  \BibitemOpen
  \bibfield  {author} {\bibinfo {author} {\bibfnamefont {D.}~\bibnamefont
  {Schmid}}\ and\ \bibinfo {author} {\bibfnamefont {R.~W.}\ \bibnamefont
  {Spekkens}},\ }\href@noop {} {\bibfield  {journal} {\bibinfo  {journal}
  {Physical Review X}\ }\textbf {\bibinfo {volume} {8}},\ \bibinfo {pages}
  {011015} (\bibinfo {year} {2018})}\BibitemShut {NoStop}%
\bibitem [{\citenamefont {Lostaglio}\ and\ \citenamefont
  {Senno}(2020)}]{lostaglio2020contextual}%
  \BibitemOpen
  \bibfield  {author} {\bibinfo {author} {\bibfnamefont {M.}~\bibnamefont
  {Lostaglio}}\ and\ \bibinfo {author} {\bibfnamefont {G.}~\bibnamefont
  {Senno}},\ }\href@noop {} {\bibfield  {journal} {\bibinfo  {journal}
  {Quantum}\ }\textbf {\bibinfo {volume} {4}},\ \bibinfo {pages} {258}
  (\bibinfo {year} {2020})}\BibitemShut {NoStop}%
\bibitem [{\citenamefont {Lostaglio}(2020)}]{lostaglio2020certifying}%
  \BibitemOpen
  \bibfield  {author} {\bibinfo {author} {\bibfnamefont {M.}~\bibnamefont
  {Lostaglio}},\ }\href@noop {} {\bibfield  {journal} {\bibinfo  {journal}
  {Physical Review Letters}\ }\textbf {\bibinfo {volume} {125}},\ \bibinfo
  {pages} {230603} (\bibinfo {year} {2020})}\BibitemShut {NoStop}%
\bibitem [{\citenamefont {Gao}\ \emph {et~al.}(2021)\citenamefont {Gao},
  \citenamefont {Anschuetz}, \citenamefont {Wang}, \citenamefont {Cirac},\ and\
  \citenamefont {Lukin}}]{gao2021enhancing}%
  \BibitemOpen
  \bibfield  {author} {\bibinfo {author} {\bibfnamefont {X.}~\bibnamefont
  {Gao}}, \bibinfo {author} {\bibfnamefont {E.~R.}\ \bibnamefont {Anschuetz}},
  \bibinfo {author} {\bibfnamefont {S.-T.}\ \bibnamefont {Wang}}, \bibinfo
  {author} {\bibfnamefont {J.~I.}\ \bibnamefont {Cirac}}, \ and\ \bibinfo
  {author} {\bibfnamefont {M.~D.}\ \bibnamefont {Lukin}},\ }\href@noop {}
  {\bibfield  {journal} {\bibinfo  {journal} {arXiv preprint arXiv:2101.08354}\
  } (\bibinfo {year} {2021})}\BibitemShut {NoStop}%
\bibitem [{\citenamefont {Schmid}\ \emph
  {et~al.}(2021{\natexlab{b}})\citenamefont {Schmid}, \citenamefont {Du},
  \citenamefont {Selby},\ and\ \citenamefont {Pusey}}]{schmid2021only}%
  \BibitemOpen
  \bibfield  {author} {\bibinfo {author} {\bibfnamefont {D.}~\bibnamefont
  {Schmid}}, \bibinfo {author} {\bibfnamefont {H.}~\bibnamefont {Du}}, \bibinfo
  {author} {\bibfnamefont {J.~H.}\ \bibnamefont {Selby}}, \ and\ \bibinfo
  {author} {\bibfnamefont {M.~F.}\ \bibnamefont {Pusey}},\ }\href@noop {}
  {\bibfield  {journal} {\bibinfo  {journal} {arXiv preprint arXiv:2101.06263}\
  } (\bibinfo {year} {2021}{\natexlab{b}})}\BibitemShut {NoStop}%
\bibitem [{\citenamefont {Howard}\ \emph {et~al.}(2014)\citenamefont {Howard},
  \citenamefont {Wallman}, \citenamefont {Veitch},\ and\ \citenamefont
  {Emerson}}]{howard2014contextuality}%
  \BibitemOpen
  \bibfield  {author} {\bibinfo {author} {\bibfnamefont {M.}~\bibnamefont
  {Howard}}, \bibinfo {author} {\bibfnamefont {J.}~\bibnamefont {Wallman}},
  \bibinfo {author} {\bibfnamefont {V.}~\bibnamefont {Veitch}}, \ and\ \bibinfo
  {author} {\bibfnamefont {J.}~\bibnamefont {Emerson}},\ }\href@noop {}
  {\bibfield  {journal} {\bibinfo  {journal} {Nature}\ }\textbf {\bibinfo
  {volume} {510}},\ \bibinfo {pages} {351} (\bibinfo {year}
  {2014})}\BibitemShut {NoStop}%
\bibitem [{\citenamefont {Raussendorf}(2013)}]{raussendorf2013contextuality}%
  \BibitemOpen
  \bibfield  {author} {\bibinfo {author} {\bibfnamefont {R.}~\bibnamefont
  {Raussendorf}},\ }\href@noop {} {\bibfield  {journal} {\bibinfo  {journal}
  {Physical Review A}\ }\textbf {\bibinfo {volume} {88}},\ \bibinfo {pages}
  {022322} (\bibinfo {year} {2013})}\BibitemShut {NoStop}%
\bibitem [{\citenamefont {Mansfield}\ and\ \citenamefont
  {Kashefi}(2018)}]{mansfield2018quantum}%
  \BibitemOpen
  \bibfield  {author} {\bibinfo {author} {\bibfnamefont {S.}~\bibnamefont
  {Mansfield}}\ and\ \bibinfo {author} {\bibfnamefont {E.}~\bibnamefont
  {Kashefi}},\ }\href@noop {} {\bibfield  {journal} {\bibinfo  {journal}
  {Physical review letters}\ }\textbf {\bibinfo {volume} {121}},\ \bibinfo
  {pages} {230401} (\bibinfo {year} {2018})}\BibitemShut {NoStop}%
\bibitem [{\citenamefont {Abramsky}\ \emph {et~al.}(2017)\citenamefont
  {Abramsky}, \citenamefont {Barbosa},\ and\ \citenamefont
  {Mansfield}}]{abramsky2017contextual}%
  \BibitemOpen
  \bibfield  {author} {\bibinfo {author} {\bibfnamefont {S.}~\bibnamefont
  {Abramsky}}, \bibinfo {author} {\bibfnamefont {R.~S.}\ \bibnamefont
  {Barbosa}}, \ and\ \bibinfo {author} {\bibfnamefont {S.}~\bibnamefont
  {Mansfield}},\ }\href@noop {} {\bibfield  {journal} {\bibinfo  {journal}
  {Physical review letters}\ }\textbf {\bibinfo {volume} {119}},\ \bibinfo
  {pages} {050504} (\bibinfo {year} {2017})}\BibitemShut {NoStop}%
\bibitem [{\citenamefont {Shahandeh}(2021)}]{shahandeh2021contextuality}%
  \BibitemOpen
  \bibfield  {author} {\bibinfo {author} {\bibfnamefont {F.}~\bibnamefont
  {Shahandeh}},\ }\href@noop {} {\bibfield  {journal} {\bibinfo  {journal} {PRX
  Quantum}\ }\textbf {\bibinfo {volume} {2}},\ \bibinfo {pages} {010330}
  (\bibinfo {year} {2021})}\BibitemShut {NoStop}%
\bibitem [{\citenamefont {Galve}\ \emph {et~al.}(2016)\citenamefont {Galve},
  \citenamefont {Zambrini},\ and\ \citenamefont
  {Maniscalco}}]{Galve__NMarkovQD_2016}%
  \BibitemOpen
  \bibfield  {author} {\bibinfo {author} {\bibfnamefont {F.}~\bibnamefont
  {Galve}}, \bibinfo {author} {\bibfnamefont {R.}~\bibnamefont {Zambrini}}, \
  and\ \bibinfo {author} {\bibfnamefont {S.}~\bibnamefont {Maniscalco}},\
  }\href {\doibase 10.1038/srep19607} {\bibfield  {journal} {\bibinfo
  {journal} {Scientific Reports}\ }\textbf {\bibinfo {volume} {6}} (\bibinfo
  {year} {2016}),\ 10.1038/srep19607}\BibitemShut {NoStop}%
\bibitem [{\citenamefont {Oliveira}\ \emph {et~al.}(2019)\citenamefont
  {Oliveira}, \citenamefont {de~Paula},\ and\ \citenamefont
  {Drumond}}]{Sheilla_QDNMarkov_2019}%
  \BibitemOpen
  \bibfield  {author} {\bibinfo {author} {\bibfnamefont {S.~M.}\ \bibnamefont
  {Oliveira}}, \bibinfo {author} {\bibfnamefont {A.~L.}\ \bibnamefont
  {de~Paula}}, \ and\ \bibinfo {author} {\bibfnamefont {R.~C.}\ \bibnamefont
  {Drumond}},\ }\href {\doibase 10.1103/physreva.100.052110} {\bibfield
  {journal} {\bibinfo  {journal} {Physical Review A}\ }\textbf {\bibinfo
  {volume} {100}} (\bibinfo {year} {2019}),\
  10.1103/physreva.100.052110}\BibitemShut {NoStop}%
\bibitem [{\citenamefont {Duarte}\ and\ \citenamefont
  {Amaral}(2018)}]{duarte2018resource}%
  \BibitemOpen
  \bibfield  {author} {\bibinfo {author} {\bibfnamefont {C.}~\bibnamefont
  {Duarte}}\ and\ \bibinfo {author} {\bibfnamefont {B.}~\bibnamefont
  {Amaral}},\ }\href {\doibase 10.1063/1.5018582} {\bibfield  {journal}
  {\bibinfo  {journal} {Journal of Mathematical Physics}\ }\textbf {\bibinfo
  {volume} {59}},\ \bibinfo {pages} {062202} (\bibinfo {year}
  {2018})}\BibitemShut {NoStop}%
\bibitem [{\citenamefont {Wagner}\ \emph {et~al.}(2021)\citenamefont {Wagner},
  \citenamefont {Baldijão}, \citenamefont {Tezzin},\ and\ \citenamefont
  {Amaral}}]{wagner2021using}%
  \BibitemOpen
  \bibfield  {author} {\bibinfo {author} {\bibfnamefont {R.}~\bibnamefont
  {Wagner}}, \bibinfo {author} {\bibfnamefont {R.~D.}\ \bibnamefont
  {Baldijão}}, \bibinfo {author} {\bibfnamefont {A.}~\bibnamefont {Tezzin}}, \
  and\ \bibinfo {author} {\bibfnamefont {B.}~\bibnamefont {Amaral}},\
  }\href@noop {} {\enquote {\bibinfo {title} {Using a resource theoretic
  perspective to witness and engineer quantum generalized contextuality for
  prepare-and-measure scenarios},}\ } (\bibinfo {year} {2021}),\ \Eprint
  {http://arxiv.org/abs/2102.10469} {arXiv:2102.10469 [quant-ph]} \BibitemShut
  {NoStop}%
\bibitem [{\citenamefont {Baumgartner}(2011)}]{baumgartner2011inequality}%
  \BibitemOpen
  \bibfield  {author} {\bibinfo {author} {\bibfnamefont {B.}~\bibnamefont
  {Baumgartner}},\ }\href@noop {} {\enquote {\bibinfo {title} {An inequality
  for the trace of matrix products, using absolute values},}\ } (\bibinfo
  {year} {2011}),\ \Eprint {http://arxiv.org/abs/1106.6189} {arXiv:1106.6189
  [math-ph]} \BibitemShut {NoStop}%
\bibitem [{\citenamefont {Gallier}()}]{Gallier_UPennNotes}%
  \BibitemOpen
  \bibfield  {author} {\bibinfo {author} {\bibfnamefont {J.~H.}\ \bibnamefont
  {Gallier}},\ }\href {https://www.cis.upenn.edu/~cis610/geombchap3.pdf}
  {\enquote {\bibinfo {title} {Properties of convex sets: A glimpse},}\
  }\bibinfo {note} {Lecture Notes}\BibitemShut {NoStop}%
\bibitem [{\citenamefont {Bergou}\ \emph {et~al.}(2004)\citenamefont {Bergou},
  \citenamefont {Herzog},\ and\ \citenamefont
  {Hillery}}]{Bergou_StateDiscr_2004}%
  \BibitemOpen
  \bibfield  {author} {\bibinfo {author} {\bibfnamefont {J.~A.}\ \bibnamefont
  {Bergou}}, \bibinfo {author} {\bibfnamefont {U.}~\bibnamefont {Herzog}}, \
  and\ \bibinfo {author} {\bibfnamefont {M.}~\bibnamefont {Hillery}},\
  }\enquote {\bibinfo {title} {11 discrimination of quantum states},}\ in\
  \href {\doibase 10.1007/978-3-540-44481-7_11} {\emph {\bibinfo {booktitle}
  {Quantum State Estimation}}},\ \bibinfo {editor} {edited by\ \bibinfo
  {editor} {\bibfnamefont {M.}~\bibnamefont {Paris}}\ and\ \bibinfo {editor}
  {\bibfnamefont {J.}~\bibnamefont {{\v{R}}eh{\'a}{\v{c}}ek}}}\ (\bibinfo
  {publisher} {Springer Berlin Heidelberg},\ \bibinfo {address} {Berlin,
  Heidelberg},\ \bibinfo {year} {2004})\ pp.\ \bibinfo {pages}
  {417--465}\BibitemShut {NoStop}%
\bibitem [{\citenamefont {Gitton}\ and\ \citenamefont
  {Woods}(2021)}]{gitton2021solvable}%
  \BibitemOpen
  \bibfield  {author} {\bibinfo {author} {\bibfnamefont {V.}~\bibnamefont
  {Gitton}}\ and\ \bibinfo {author} {\bibfnamefont {M.~P.}\ \bibnamefont
  {Woods}},\ }\href@noop {} {\enquote {\bibinfo {title} {Solvable criterion for
  the contextuality of any prepare-and-measure scenario},}\ } (\bibinfo {year}
  {2021}),\ \Eprint {http://arxiv.org/abs/2003.06426} {arXiv:2003.06426
  [quant-ph]} \BibitemShut {NoStop}%
\end{thebibliography}%
\appendix
\section{Formal presentation and discussion of the results of Brand\~{a}o, Piani and Horodecki}
\label{App: BPHandRelax}

\subsection{Brand\~{a}o, Piani and Horodecki's results}
\label{App: BPH}
In this section, for the benefit of the reader, we first present one of Brand\~{a}o, Piani and Horodecki's results in rigorous form to then discuss quickly the consequences of our assumption of infinite subsystems in the environment.\\

Recall the EW$_t$-dynamics: there is an environment with $N$ subsystems $\{B_1,\ldots,B_N\}$, and we want to focus on a portion with $t$ individual subsystems, $B_{S_t}$, of this environment. Theorem $1$ of the main text, proved in reference \cite{BPH_2015_QD} and responsible for objectivity of observables in this dynamics, can be rigorously put  as follows:

\begin{theorem}[Theorem $2$ in ref. \cite{BPH_2015_QD}]
Let $\Phi^{S_t}:\mathcal{D}(\mathcal{H}_A) \to \mathcal{D}(\bigotimes_{j\in S_t}\mathcal{H}_{B_j})$ be an EW$_t$-dynamics, where $S_t \subset \{1,\dots,N\}$. For every $0<\delta<1$ there exists a POVM $\{\Tilde{E}_k\}_k$ such that for more than a $(1-\delta)$ fraction of the subsets $S_t$,
\begin{equation}
    \left\Vert \Phi^{S_t} - \Phi^{S_t}_{obs}\right\Vert_{\diamond}\leq \left(\frac{27 \ln(2)d_A^6\log(d_A)t}{N\delta^3}\right)^{\frac{1}{3}},
\end{equation}
with $d_A \equiv \mathrm{dim}(\mathcal{H}_A)$, and where $\Phi_{obs}^{S_t}$ is a measure-and-prepare map with respect to the family of states $(\sigma_k^{S_t})_k$, meaning that for all $\rho \in \mathcal{D}(\mathcal{H}_A)$, 
\begin{equation}
    \Phi_{obs}^{S_t}(\rho) = \sum_k {\rm Tr}\{\tilde{E}_k\rho\}\sigma_k^{S_t}.
\end{equation}

\end{theorem}

In the above theorem, $\delta$ quantifies how many fractions of $t$ portions of the environment the approximation will hold. From the upper bound given above, it is possible to see that, if $\delta$ is close enough to $1$, our hypothesis of $S_t$ being one such fraction is not too restrictive. Next, if this value $\delta$ is fixed (or even increase slowly enough with $N$), when we take the limit $N \to \infty$ we have $\Phi^{S_t} = \Phi^{S_t}_{obs}$, as it was considered in the main text. Since environments are usually considered to have a very large quantity of subsystems, this hypothesis is well-motivated. 

\subsection{Relaxing the assumption of infinite environment}
\label{App: Relax}
In this subsection we relax the assumption of infinite-sized environments, allowing for a deviation $\Phi^{B_{S_t}}\neq \Phi^{B_{S_t}}_{\rm obs}$. We show that, as long as we consider a finite set of procedures, outcomes and equivalences, objectivity of observables hardly constrain contextuality, as measured by the quantifier based on the $l_1$-distance \cite{duarte2018resource}.

In what follows, we assume that the number of procedures and outcomes are finite, as well as $\text{Equiv}(\mathcal{P'}), \text{Equiv}(\mathcal{M}^{B_j})$ -- which match the features of any realist experiment.
This means that the effective prepare-and-measure contextuality scenario for each Bob $B_j$ can be defined as the tuple
 $$\mathfrak{S}_j\equiv\left(T_j(\mathcal{P}), \mathcal{M}^{B_j}, \mathcal{O}_{\mathcal{M}^{B_j}},\text{Equiv}(T_j(\mathcal{P})),\text{Equiv}(\mathcal{M}^{B_j})\right)$$ where these sets are as defined in the main text and where the set $\mathcal{O}_{\mathcal{M}^{B_j}}$ is the set of outcome labels of Bob $B_j$. Given that all these sets are finite they define prepare-and-measure scenarios as in ref.~\cite{schmid2018all}. Let $C\left(\mathfrak{S}_j\right)$ represent the set of all possible behaviours, $p \equiv (p(b|M,P'))_{b\in \mathcal{O}_{\mathcal{M}^{B_j}},M\in\mathcal{M}^{B_j},P'\in \mathcal{P}'}$ of $\mathfrak{S}_j$; in the case of finitely many procedures, equivalences and outcomes, this forms a polytope \cite{schmid2018all}. We write the inner set of noncontextual behaviours (also a polytope in the finite scenario) as $NC\left(\mathfrak{S}_j\right)$. We consider the resource theory framework constructed for these scenarios \cite{duarte2018resource,wagner2021using}, in particular, the definition of the $l_1$-contextuality distance.
\begin{equation*}
    \mathbf{d}(p) := \min_{q\in NC\left(\mathfrak{S}_j\right)} \max_{\begin{subarray}{l} M\in \mathcal{M}^{B_j}\\P'\in \mathcal{P}'\end{subarray}} \sum_b \vert p(b|M,P') - q(b|M,P')\vert.
\end{equation*}
For any particular noncontextual behaviour $q_{*}$,

\begin{equation}
    \mathbf{d}(p) \leq \max_{\begin{subarray}{l} M\in \mathcal{M}^{B_j}\\P'\in \mathcal{P}'\end{subarray}} \sum_b \vert p(b|M,P') - q_{*}(b|M,P')\vert.
\end{equation}
Recall that, in particular, when $(\sigma_k)_k^{B_j}$ constitutes a family of affinely independent states we will have that $q_{*} = \left(\Tr\{\Phi^{B_j}_{obs}(\rho)F_b^M\}\right)_{ b\in \mathcal{O}_{\mathcal{M}^{B_j}}, M\in \mathcal{M}^{B_j}, P' \in \mathcal{P}'}$ is a noncontextual behaviour (lemma \ref{lemma: AIimplyNC} the main text ), where we are using POVM notation $\{F_b^M\}_b = M$. Since $\mathcal{P}' = T_j(\mathcal{P})$ we can use the labels for states in $\mathcal{P}$, meaning that Alice prepares quantum realizations as the states $\rho_P$ for any $P \in \mathcal{P}$ that are later transformed before reaching Bob $B_j$. If $p$ is any quantum behaviour in this scenario we must have that

\begin{align*}
    &\mathbf{d}(p) \leq \max_{\begin{subarray}{l} M\in \mathcal{M}^{B_j}\\P\in \mathcal{P}\end{subarray}} \sum_b \left\vert \Tr\{F_b^M \Phi^{B_j}(\rho_P)\} - \Tr\{F_b^M\Phi^{B_j}_{obs}(\rho_P)\} \right\vert \\
    &= \max_{\begin{subarray}{l} M\in \mathcal{M}^{B_j}\\P\in \mathcal{P}\end{subarray}} \sum_b \left\vert \Tr\{F_b^M (\Phi^{B_j}-\Phi^{B_j}_{obs})(\rho_P)\} \right\vert \\
    &=  \max_{\begin{subarray}{l} M\in \mathcal{M}^{B_j}\\P\in \mathcal{P}\end{subarray}}\sum_b \left\vert \Tr\{F_b^M  \mathcal{R}_j(\rho_P)\} \right\vert. 
\end{align*}
Where $\mathcal{R}_j := \Phi^{B_j}-\Phi^{B_j}_{obs}$ for brevity. Notice that $\mathcal{R}_j(\rho_P)$ is a self-adjoint operator for any $\rho_P$ in the domain. To proceed, recall the following  inequality \cite{baumgartner2011inequality},  

\begin{lemma}[Hölder's inequality]
Let $A,B$ be any $n\times n$ complex matrices. Then, 
\begin{equation}
    \left\vert \Tr{A^\dagger B} \right \vert \leq (\Tr{\vert A\vert^l})^{\frac{1}{l}}(\Tr{\vert B\vert^s} )^{\frac{1}{s}}
\end{equation}
such that $1\leq l,s\leq \infty$ with $\frac{1}{l} + \frac{1}{s} = 1$. 
\end{lemma}

Here, we fix the notation $\vert A \vert := \sqrt{A^\dagger A}$, while for scalars $x$ the same symbol $|x|$ means the absolute value. From Hölder's inequality we have that,

\begin{equation*}
    \left \vert \Tr\{F^{M}_b \mathcal{R}_j(\rho_P)\} \right \vert \leq \sqrt{\Tr\{\vert \mathcal{R}_j(\rho_P)\vert^2 \}}\sqrt{\Tr\{\vert F_b^{M}\vert^2 \}}
\end{equation*}
and we get,
\begin{align*}
    \mathbf{d}(p) &\leq \max_{\begin{subarray}{l} M\in \mathcal{M}^{B_j}\\P\in \mathcal{P}\end{subarray}}\sum_b \left\vert \Tr\{F_b^M \mathcal{R}_j(\rho_P)\} \right\vert\\
    &\leq\max_{\begin{subarray}{l} M\in \mathcal{M}^{B_j}\\P\in \mathcal{P}\end{subarray}} \sum_b  \sqrt{\Tr\{\vert \mathcal{R}_j(\rho_P)\vert^2 \}}\sqrt{\Tr\{\vert F_b^{M}\vert^2 \}} \\
    &= \max_{M\in \mathcal{M}^{B_j}}\sum_b\sqrt{\Tr\{\vert F_b^{M}\vert^2 \}}\max_{P\in \mathcal{P}} \sqrt{\Tr\{\vert \mathcal{R}_j(\rho_P)\vert^2 \}}\\
    &= C\max_{P\in \mathcal{P}} \sqrt{\Tr\{\vert \mathcal{R}_j(\rho_P)\vert^2 \}}
\end{align*}
where $C:=  \max_{M\in \mathcal{M}^{B_j}}\sum_b\sqrt{\Tr\{\vert F_b^{M}\vert^2 \}}$ is a non-zero bounded constant for any set $\mathcal{M}^{B_j}$. For IC-POVMs $0<C \leq d_{B_j}^3$. Generally speaking $0<C \leq d_{B_j} \# \mathcal{O}_{\mathcal{M}^{B_j}}$, the dimension of Bob's system $d_{B_j}$ and the number of outcomes. We can further improve this inequality with the following lemma,

\begin{lemma}
\label{lemma: desigualdade importante} Let $A \in \mathcal{B}(\mathcal{H})$, with $\dim(\mathcal{H}) < \infty$. Then,  $\sqrt{\Tr{|A|^2}} \leq \Tr{|A|}.$
\end{lemma}

\begin{proof} Calling $G = A^\dagger A$  we have $\sqrt{G}^2 = \sum_{\lambda\in \sigma(G)} \sqrt{\lambda}^2P_\lambda = G$, with $\sigma(G)$ the spectrum of $G$ and $P_\lambda$ spectral projectors. Hence, $\sqrt{\Tr{|A|^2}} = \sqrt{\sum_{\lambda \in \sigma(G)} \lambda} \leq \sum_{\lambda \in \sigma(G)}\sqrt{\lambda} = \sum_{\lambda \in \sigma\left(\sqrt{G}\right)}\lambda = \Tr{|A|}$ since $\sqrt{G} = |A|$.
\end{proof}
Using this lemma we show that
\begin{align*}
    \mathbf{d}(p) &\leq C\max_{P\in \mathcal{P}} \sqrt{\Tr\{\vert \mathcal{R}_j(\rho_P)\vert^2 \}}\\
    &\stackrel{\text{Lemma } \ref{lemma: desigualdade importante}}{\leq} C\max_{P\in \mathcal{P}} \Tr\{\vert \mathcal{R}_j(\rho_P)\vert\},
\end{align*}
and therefore
\begin{align*}
   \mathbf{d}(p) &\leq C \max_{\rho \in \mathcal{D}(\mathcal{H}_A)}\Tr\{\vert \mathcal{R}_j(\rho)\vert\}\\
    &= C \max_{\rho \in \mathcal{D}(\mathcal{H}_A)} \left\Vert (\Phi^{B_j}-\Phi^{B_j}_{obs})(\rho)\right\Vert_1 \\
    &= C \max_{\rho \in \mathcal{D}(\mathcal{H}_A)} \left\Vert \frac{1}{d_A}\text{id}_A(\mathbb{1}_A)\otimes  \left((\Phi^{B_j}-\Phi^{B_j}_{obs})(\rho)\right)\right\Vert_1
\end{align*}
where the last equality comes from the fact that the singular values of a matrix that is the tensor product is simply the product of the singular values. Our notation is  $\mathbb{1}_A \in \mathcal{D}(\mathcal{H}_A)$ and $\text{id}_A: \mathcal{D}(\mathcal{H}_A) \to \mathcal{D}(\mathcal{H}_A)$ with $\text{id}_A(\rho) = \rho$. Continuing with the manipulations we finally find, defining $\mathcal{H}_A^{\otimes 2} \equiv \mathcal{H}_A \otimes \mathcal{H}_A$,  

\begin{align*}
    \mathbf{d}(p) &\leq \frac{C}{d_A} \max_{\mathbb{1}_A\otimes \rho \in \mathcal{D}(\mathcal{H}_A^{\otimes 2})} \left\Vert \text{id}_A\otimes  \left(\Phi^{B_j}-\Phi^{B_j}_{obs}\right)(\mathbb{1}_A \otimes \rho)\right\Vert_1 \\
    &\leq \frac{C}{d_A} \max_{\sigma \in \mathcal{D}(\mathcal{H}_A^{\otimes 2})} \left\Vert \text{id}_A\otimes  \left(\Phi^{B_j}-\Phi^{B_j}_{obs}\right)(\sigma )\right\Vert_1 \\
    &= \frac{C}{d_A} \left\Vert \Phi^{B_j}-\Phi^{B_j}_{obs} \right\Vert_{\diamond}.
\end{align*}
The relation found above tells us that contextuality, as quantified by the $l_1$-distance, is bounded by the distance of $\Phi^{B_{S_t}}$ to $\Phi^{B_{S_t}}_{\rm obs}$, as measured by the diamond norm. Now, Brand\~{a}o, Piani and Horodecki's theorem reproduced in the previous subsection tells us that this distance is bounded, showing that the EW$_t$-dynamics bounds contextuality even in the case of finite-sized environment.

As an important instance, let Bob $B_j$ perform IC-POVM's, such that $C\leq d_{B_j}^3$. Our conclusion is that for a sufficiently large number $N$ with respect to $d_{B_j},d_A,t,\delta^{-1}$ the behaviour $p$ will be noncontextual up to precision in the estimate of $\mathbf{d}(p)$ since $\mathbf{d}(p)~\leq~\left(\frac{27 \ln(2)d_A^3d_{B_j}^9\log(d_A)t}{N\delta^3}\right)^{\frac{1}{3}} $. 

\section{Proofs of the results}
\label{App: Proofs}

The first problem to be treated here is: we have a measure-and-prepare map defined by the pairs $(\tilde{E}_k,\sigma_k)_k$ and we wish to detect affine independence of the states $(\sigma_k)_k$ in an operational task. What we do in the next subsection is to consider that, for each initial state of system $A$, $\rho^A\in\mathcal{D}(\mathcal{H}_A)$, we have an instance of a minimum-error state discrimination task. Indeed, each $\rho^A$ leads to an a priori distribution $\tilde{p}_k:={\rm Tr}[\tilde{E}_k\rho^A]$ for the states to be discriminated, $\{\sigma_k\}_k$. This leads us to a bound $\hat{P}$  such that, if $p_{\rm guess}[(\tilde{p}_k,\sigma_k)_k]>\hat{P}$ for all $\rho^A$, the states $\{\sigma_k\}_k$ must be affinely independent. 

\subsection{Proof of Lemma 3}
\label{App:lemma3}

Before proving our results, let us give a simple mathematical argument of why there should be a bound on the guessing probability separating the affinely independent and dependent cases.

Define the matrix $(W)_{ik}:= {\rm Tr}[F_i\sigma_k]$, where $\{F_i\}_i$ is a POVM that best distinguishes among the states $\{\sigma_k\}_k$. If the states are perfectly distinguishable -- in which case they must have disjoint support, thus being affinely independent-- we have $W=\mathds{1}_{k_{\rm max}\times k_{\rm max}}$. This implies ${\rm det}(W)=1$.

However, in the case where $\{\sigma_k\}_k$ is affinely dependent, the matrix $W$ must have a column which is linearly dependent on the others. This implies that ${\rm det}(W)=0$. From the continuity of the determinant, we see that there exists a ball $b$ around  $\mathds{1}$ such that $V\in b\implies {\rm det}(V)>0$, thus the underlying states defining the entries of $V$, $(\sigma_k)_k$, must be affinely independent.

With this general intuition, we follow to the specifics of our proof. First we define our distinguishability bound, $\hat{P}$, able to detect affine independence when $p_{\rm guess}>\hat{P}$.

\begin{definition}[Distinguishability Bound]
\label{def:newDistBound}
Consider a measure-and-prepare channel defined by $(\tilde{E}_k,\sigma_k)_k$, with $\tilde{E}_k\neq0$ for all $k$. Then, there will be states $\rho^A\in\mathcal{D}(\mathcal{H}_A)$ such that ${\rm Tr}[\tilde{E}_k\rho^A]\neq 0$ for all $k$. Denote the set of such states by $\mathcal{S}$. Now, assume (w.l.g.) that ${\rm Tr}[\tilde{E}_1\rho^A]\geq{\rm Tr}[\tilde{E}_2\rho^A]\geq\ldots\geq {\rm Tr}[\tilde{E}_{k_{\rm max}}\rho^A]>0$ (otherwise, relabel $\{\tilde{E}_k\}$ so that it does).  We define the \textbf{distinguishability bound} $\hat{P}$ as

\begin{align}
    \hat{P}[(\tilde{E}_k)_k] &:= \min_{\rho^A\in \mathcal{S}}\left[\sum_{i=1}^{k_{\rm max}-1} {{\rm Tr}[\tilde{E}_k\rho^A]}+\frac{{{\rm Tr}[\tilde{E}_{k_{\rm max}}\rho^A]}}{2}\right]\nonumber\\
    &= 1-\frac{1}{2}\max_{\rho^A\in\mathcal{S}}{\rm Tr}[\tilde{E}_{k_{\rm max}}\rho^A].
    \label{eq: NewDistBound}
\end{align}
\end{definition}
Note that, indeed, $\mathcal{S}$ is not empty: since $\tilde{E}_k\neq0$ for all $k$, at least states of the form $(1-a)\rho^A+a(\mathds{1}/d_A)$, with $a\neq 0$, belong to this set.
In what follows and for simplicity, we will always consider that $(\tilde{p}_k)_k = \left(\Tr[\tilde{E}_k\rho^A]\right)_k$ is non-increasing, so $\tilde{p}_{k_{\rm max}}$ is always the smallest value of the distribution $(\tilde{p}_k)_k$ generated by a state $\rho^A\in\mathcal{S}$, as in definition \ref{def:newDistBound}. There is no loss of generality here, since equivalent proofs can be written by relabelling and/or changing $\tilde{p}_{k_{\rm max}}$ to $\min_k[(\tilde{p}_k)_k]$.   
We will take advantage of the fact that, if $\{\sigma_k\}_k$ are affinely dependent, they \emph{do not} form vertices of a $(k_{\rm max}-1)$-simplex, thus interior points will have non-unique convex decompositions in terms of these states and this will make it impossible to have $p_{\rm guess}>\hat{P}$ for all initial states $\rho^A$. We will actually consider a stronger consequence of affine dependence, captured by Carathéodory's theorem: we can always describe interior points using convex combinations of $\{\sigma_k\}_k$ with some null coefficients (see Fig.~\ref{fig:Caratheodory}).

\begin{theorem}[Carathéodory, adapted from version of ref. \cite{Gallier_UPennNotes}]Given any affine space $E$ of dimension $n$, for any (non-void) family  $f=\{\sigma_k\}_{k=1}^{k_{\rm max}}$ in $E$, the set ${\rm ConvHull}[f]$ is equal to the set of convex combinations of families of $n+1$ points of $f$.
\end{theorem}

Let us interpret this result in our terms.
The affine space $E:={\rm AffineHull}[\{\sigma_k\}_k]$ will have dimension $n\leq k_{\rm max}-1$, where equality is reached iff $\{\sigma_k\}_k$ is an affinely independent set. Carathéodory's theorem thus says that any point $\sigma\in {\rm ConvHull}[\{\sigma_k\}_{k=1}^{k_{\rm max}}]$ can be written as a convex combination of at most $n+1$ points of $\{\sigma_k\}_k$.

\begin{figure}[hb]
    \centering
    \includegraphics[width=\columnwidth]{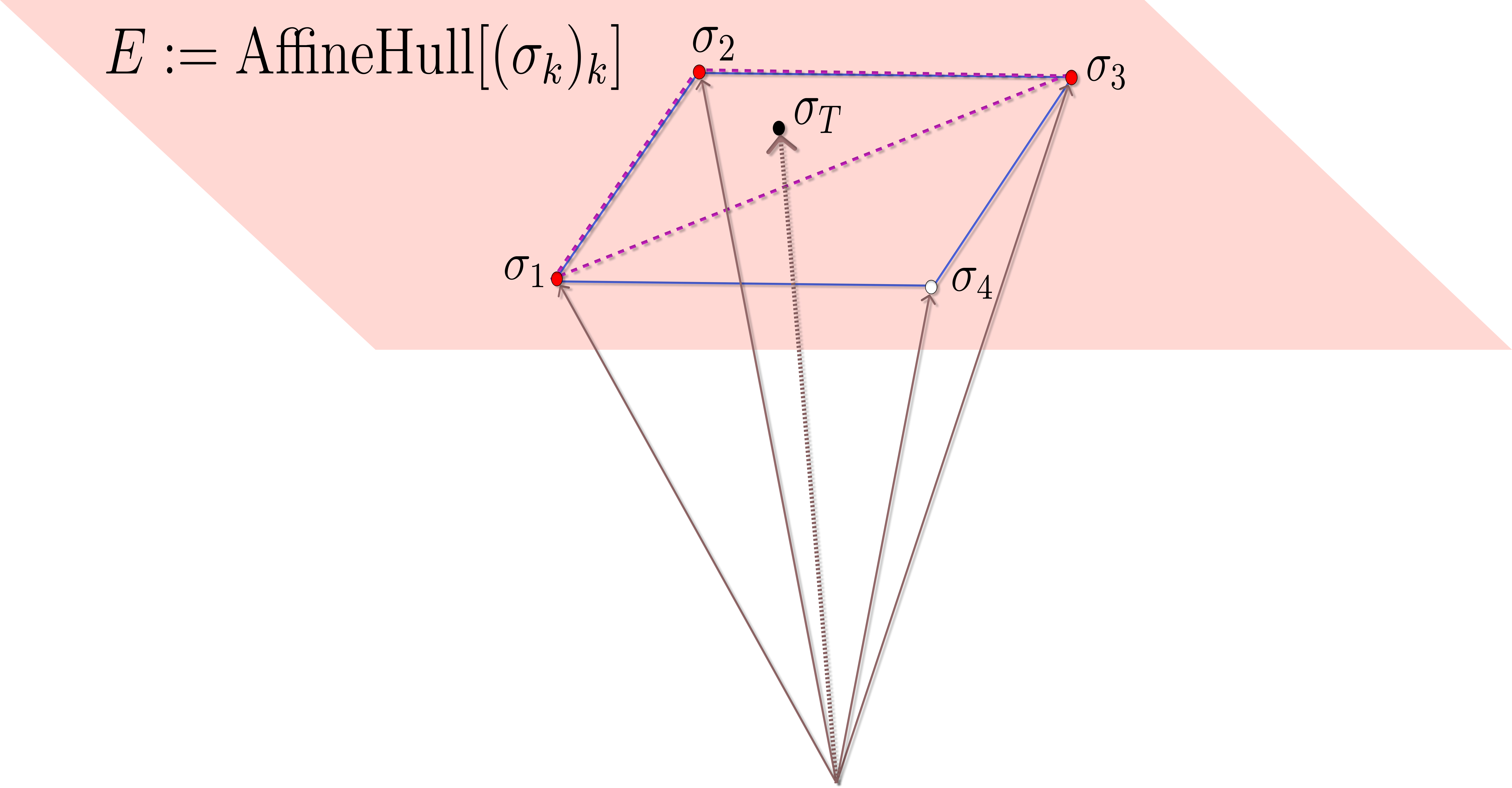}
    \caption{\textbf{Example of Caratheodory's Theorem:} an example of an affinely dependent set with $4$ states ($k_{\rm max}=4$). Their affine hull $E$ forms a plane, thus $\dim(E)=2$, and ${\rm ConvHull}[\{\sigma_k\}_k]$ is a quadrilateral contained in $E$. An interior point, $\sigma_T$, can be written as a convex combination of $2+1=3$ vertices (as shown by the red vertices and the purple dashed triangle). Thus, we can write $\sigma_T=\sum_k q_k \sigma_k$ with $q_4=0$.}
    \label{fig:Caratheodory}
\end{figure}

Now, suppose $\{\sigma_k\}_k$ \emph{does not} form an affinely independent set. Then, the dimension of the affine space $E$ is strictly less than $k_{\rm max}-1$, i.e.  ${\rm dim}(E)\leq k_{\rm max}-2$. From Carathéodory's theorem, any point of ${\rm ConvHull}[\{\sigma_k\}_k]$ can thus be written as a convex combination of $k_{\rm max}-1$ vertices. In particular, consider a density matrix $\sigma_T=\sum_k\tilde{p}_k\sigma_k$  with $\tilde{p}_k>0$ for all $k$; we can be sure that there exists another set of convex coefficients, $\{q_k\}$, such that $q_{k}=0$ for at least one $k$ (there may be more zeros if ${\rm dim}(E)<k_{\rm max}-2$). Formalizing:
\begin{corollary}[]
\label{corollary:qk0}
Consider any state of the form $\sigma_T=\sum_k\tilde{p}_k\sigma_k$ with $\tilde{p}_k>0$ for $k\in\{1,\ldots,k_{\rm max}\}$. Then, if $(\sigma_k)_k$ is an affinely dependent set, there exists a set of convex coefficients $\{q_k\}_k$, with $q_k=0$ for at least one value of $k$, such that $\sigma_T=\sum_k q_k\sigma_k$.
\end{corollary}

The comparison of $\{q_k\}$ and $\{\tilde{p}_k\}$ will have important consequences to us. We might quantify it via `the' statistical distance.

\begin{remark}
\label{remark:DistancePQ}
Assume $\{\sigma_k\}_k$ to be an affinely dependent set. Consider a state $\sigma_T=\sum_k\tilde{p}_k\sigma_k=\sum_k q_k\sigma_k$ with $\tilde{p}_1\geq \tilde{p}_2\ldots\geq \tilde{p}_{k_{\rm max}}>0$ and $q_k=0$ for some $k$. Define the statistical distance between $\{q_k\}_k$ and $\{\tilde{p}_k\}_k$ as
\begin{align}
    D(\{q_k\})=\frac{1}{2}\sum_k |q_k-\tilde{p}_k|.
\end{align}
Then, $D(\{\tilde{p}_k\},\{q_k\})\geq \tilde{p}_{k_{\rm max}}$.
\end{remark}
Let us justify this remark. First, define the set $K:=~\{k|q_{k}\geq \tilde{p}_k\}$ and suppose $q_{k_{\rm max}}=0$. Now, note that we can rewrite the statistical distance:
 \begin{align*}
     D(\{q_k\},\{\tilde{p}_k\}) &= \frac{1}{2}\sum_k |q_k-\tilde{p}_k| \\&=
      \frac{1}{2}\left(\sum_{k\in K}(q_k-\tilde{p}_k) + \sum_{k\notin K}(\tilde{p}_k-q_k)\right)\\
     &=\sum_{k\in K} (q_k-\tilde{p}_k).
  \end{align*}
 The last step uses the fact that $\sum_{k \notin K}q_k =1-\sum_{k\in K}q_k$ and similarly for $\tilde{p}_k$, so we can substitute the second summation by another over $k\in K$ and the last equation follows. Now, we assume $\tilde{p}_{k_{\rm max}}>0$, while $q_{k_{\rm max}}=0$, so $k_{\rm max}\notin K$. Thus,
\begin{align*}
D(\{q_k\},\{\tilde{p}_k\})&=\sum_{k\in K} (q_k-\tilde{p}_k)\\
&\geq \sum_{k\neq k_{\rm max}}(q_k-\tilde{p}_k)\\
&=1-\sum_{k\neq k_{\rm max}}\tilde{p}_k=\tilde{p}_{k_{max}}.
\end{align*}
The inequality comes from the fact that we might have included some values of $k\notin K$, which can only decrease the summation. Now, suppose $q_{k_{\rm max}}\neq 0$; we know by corollary \ref{corollary:qk0} that there will be at least one value of $k$, call it $k^*$, such that $q_{k^*}=0$. Then, by repeating the calculation above we find $D(\{q_k\},\{\tilde{p}_k\})>\tilde{p}_{k^*}\geq \tilde{p}_{k_{\rm max}}$, justifying our remark.

Finally, we have the ingredients to prove our lemma:
\begin{lemma}[Violation of $p_{\rm guess}\leq\hat{P}$ for all $\rho^A$ implies $\{\sigma_k\}_k$ is an affinely independent set]
Consider an EW dynamics defined by the measure-and-prepare channel $(\tilde{E}_k,\sigma_k){_k}$, with {$\tilde{E}_k\neq0$} for all $k$. Consider the associated distinguishability bound $\hat{P}[({\tilde{E}}_k)_k]$ as per definition \ref{def:newDistBound}.
Then, if for all $\rho^A\in\mathcal{D}(\mathcal{H}_A)$ the inequality
\begin{equation}
    p_{\rm guess}[({\rm Tr}\{\tilde{E}_k\rho^A\},\sigma_k)_k] > \hat{P}[({\tilde{E}}_k)_k] 
    \label{ineq: violationBound}
\end{equation}
holds, the states $\{\sigma_k\}_k$ are affinely independent. 
\begin{proof}

Fix $\rho^A=\bar{\rho}$, where $\bar{\rho}$ is a state attaining the minimum in the definition of the bound $\hat{P}$. Since $\bar{\rho}\in \mathcal{S}$, $\tilde{p}_k$ does not increase with $k$ and $\tilde{p}_{k_{\rm max}}>0$. Let us denote the POVM that maximizes the probability of guessing (for the specific $\rho^A\in\mathcal{S}$) as $\{F_k^*\}_k$. Now, if ${\rm Tr}[F^*_b\sigma_b]\leq\frac{1}{2}$ for some value $b$, the bound cannot be violated. Indeed,
\begin{align}
    &p_{\rm guess}[(\tilde{p}_k,\sigma_k)_k] = \sum_k\tilde{p}_k{\rm Tr}[F^*_k\sigma_k]
    \\
    &\leq\sum_{k\neq b}\tilde{p}_k\overbrace{{\rm Tr}[F^*_k\sigma_k]}^{\leq 1} +\frac{\tilde{p}_b}{2}\nonumber\\
    &\leq \sum_{k\neq b}\tilde{p}_k+\frac{\tilde{p}_b}{2}\leq \sum_{k\neq k_{\rm max}}\tilde{p}_k +\frac{\tilde{p}_{k_{\rm max}}}{2}=\hat{P},
\end{align}
(remember that $\tilde{p}_{k_{\rm max}}$ is the minimum of $(\tilde{p}_k)$, hence the last inequality. By considering the appropriate $\{F^*_k\}_k$ for each $\rho^A$, all but the last step above are valid for all $\rho^A\in\mathcal{S}$. The last equality is only valid for states attaining the minimum at the definition of the bound). So, if ${\rm Tr}[F^*_b\sigma_b]\leq 1/2$ for some $b$, we can see that $p_{\rm guess}\leq\hat{P}$ at least for $\bar{\rho}$ as the initial state.
Therefore, one needs ${\rm Tr}[E^*_k\sigma_k]~>~1/2$ for all $k$ so that $p_{\rm guess}>\hat{P}$ might be possible for all $\rho^A$. 

Now, we assume that $\{\sigma_k\}_k$ is an affinely dependent set, which leads to a contradiction with Eq.~\eqref{ineq: violationBound}. We still consider $\rho^A=\bar{\rho}\in\mathcal{S}$ and show that, if the effective average state in Bob's side ($\sigma_T=\sum_k\tilde{p}_k\sigma_k$ with $\tilde{p}_k>0$ for all $k$), admits a convex decomposition $\sigma_T=\sum_kq_k\sigma_k$ with $q_{k}=0$ for some $k$, $p_{\rm guess}<\hat{P}$ must be obeyed for $\bar{\rho}$.

The probability of error in the discrimination is defined by $P_{\rm err}=\sum_k\sum_{b\neq k}\tilde{p}_b{\rm Tr}[F^*_k\sigma_b]$. Since $\sum_{b\neq k}\tilde{p}_b\sigma_b=\sigma_T-\tilde{p}_k\sigma_k$, we can rewrite $P_{\rm err}$ as:
\begin{align}
    P_{\rm err}&=\sum_k\overbrace{{\rm Tr}[F^*_k(\sigma_T-\tilde{p}_k\sigma_k)]}^{\geq 0}\\
    &\geq\sum_{k\in K}{\rm Tr}[F^*_k(\sigma_T-\tilde{p}_k\sigma_k)]\nonumber \\
    &= \sum_{k\in K}{\rm Tr}[F^*_k(\sum_iq_i\sigma_i-\tilde{p}_k\sigma_k)]\\
    &\geq\sum_{k\in K}{\rm Tr}[F^*_k(q_k\sigma_k-\tilde{p}_k\sigma_k)].
\end{align}
In the first inequality we used the fact that the summand is positive, so restricting $k\in K$ (recall, $K:=~\{k|q_k\geq \tilde{p}_k\}$) can only decrease the summation. After substituting $\sigma_T=\sum_i q_i\sigma_i$ (second equation), we used the fact that ${\rm Tr}[F q_k\sigma_k]\leq {\rm Tr}[F \sigma_T]$ for any POVM element $F$. Therefore, we get
\begin{align}
    P_{\rm err} &\geq \sum_{k\in K}(q_k-\tilde{p}_k)\overbrace{{\rm Tr}[F^*_k\sigma_k]}^{>\frac{1}{2}}
    > \frac{1}{2}\sum_{k\in K}(q_k-\tilde{p}_k)\nonumber\\
    &=\frac{1}{2}D(\{\tilde{p}_k\},\{q_k\})\geq\frac{{p}_{k_{\rm max}}}{2},
\end{align}
where we used remark \ref{remark:DistancePQ} to replace $D(\{q_k\},\{\tilde{p}_k\})$

We can thus conclude that, if $(\sigma_k)_k$ is an affinely dependent set, $P_{\rm err}> \tilde{p}_{k_{\rm max}}/2$ for states attaining the minimum on the definition of the bound. Equivalently, $p_{\rm guess}< \sum_{k\neq k_{\rm max}}\tilde{p}_k+\tilde{p}_{k_{\rm max}}/2$ for such states, thus being impossible to obey Ineq.~\eqref{ineq: violationBound}. By contradiction, we arrive to the claim. (Note that $p_{\rm guess}$ is not bounded by the condition of affine independence, since it can get as high as $1$, when $\{\sigma_k\}_k$ are perfectly distinguishable.) 
\end{proof}
\end{lemma}

The above proof may be of interest for minimum-error state discrimination: as mentioned above, for each $\rho^A$, we have an instance of this task. Moreover, the attention we pay to states in the set $\mathcal{S}$ is very natural in minimum-error state discrimination~\cite{Bergou_StateDiscr_2004}. Indeed, in those tasks one is interested in distinguishing states $\{\sigma_k\}$, chosen according to a fixed prior $(\tilde{p}_k)$ where, usually, $\tilde{p}_k>0$ for all $k$ is assumed. Otherwise, one could just consider a smaller subset of the states $\{\sigma_k\}$ which have non-zero chance of being selected. 

Adapting to this task, our proof shows that, if $p_{\rm guess}>1-\min_k \tilde{p}_k/2$, the states $\{\sigma_k\}_k$ must be affinely independent, thus forming a $(k_{\rm max}-1)$-simplex in $\mathcal{D}(\mathcal{H})$. In other words, by observing a high value of $p_{\rm guess}$, we get an information on the geometrical disposition of the states $\{\sigma_k\}_k$.

{
\subsection{Perfect decoherence and State Spectrum Broadcasting cases}
\label{App:DecAndSSB}
\subsubsection{Dynamics with perfect decoherence}
\label{App:perfectDec}
}
Let us look into how our bound behaves in a simple and important example of an EW$_t$-dynamics: the case in which $\{\tilde{E}_k\}$ is a rank$-1$ projective pointer observable, projecting onto a basis of $\mathcal{H}_A$. This is the case of a perfect decoherence process, which is important for ideal cases of quantum Darwinism \cite{zurek_2007_RelativeQDReview}, thus forming a milestone of such processes.

In this case, as mentioned above, $\tilde{E}_k=\ket{k}\bra{k}$, and information regarding projection on this basis $\{\ket{k}\bra{k}\}_k$ may be available to the environmental states $\sigma^{B_{S_t}}$.  Initial states $\rho^A$ having the same population on this basis and differing only on the phases will lead to the same distribution ${\rm Tr}[\tilde{E}_k\rho^A]={\rm Tr}[\ket{k}\bra{k}\rho^A]$ -- as mandated by an ideal decoherence process.  Therefore, the ontological model constructed in the main text (Eq.s $8$) encapsulated this idea: $\mu_{P'}$ is the same for all initial states that decohere to the same final state. What would  our bound $\hat{P}$ be?

To answer this we need to look at the set of states $\mathcal{S}$ in definition \ref{def:newDistBound}, which contains the states having non-zero component on all of the pointer basis states. In other words, the states in $\mathcal{S}$ must obey $\tilde{p}_k=\Tr[\ket{k}\bra{k}\rho]>0$ for all $k$. Among those, we need to find the one with the highest possible value for the smaller $\tilde{p}_k$. It is possible to see that such state is the maximally mixed state: ${\rm Tr}[\ket{k}\bra{k}\rho^A]=1/{{\rm dim}(\mathcal{H}_A)}$ for all $k$. Then, in the case of perfect decoherence with a pointer \emph{basis} being selected, if
\begin{align}
    p_{\rm guess}[({\rm Tr}\{\ket{k}\bra{k}\rho^A\},\sigma_k)_k]>1-\frac{1}{2{\rm dim}(\mathcal{H}_A)}\,\,\forall \rho^A,
\end{align}
we are sure that $\{\sigma_k\}_k$ are affinely independent states. 

As we can see, for complete decoherence of systems with different dimensions, we have a different bound $\hat{P}$ -- the higher the dimension, the higher is $\hat{P}$.
This shows that choosing an arbitrary bound $\bar{\eta}$ to QD$_\eta$ processes as a threshold for objectivity can lead to problems, even though it might look `close enough to $1$'. Indeed, suppose one fixes $\bar{\eta}<1$. Now, for complete decoherence of a system with dimension ${\rm dim}(\mathcal{H}_A)>\frac{1}{2(1-\bar{\eta})}$ the bound $\hat{P}$ obeys $\hat{P}>\bar{\eta}$.

Another interesting consequence of this particular case is that we can use it to show the emergence of noncontextuality in an alternative process for objectivity in the quantum realm, namely, State Spectrum Broadcasting \cite{Horodecki_SSB_2015}.
{
\subsubsection{State Spectrum Broadcasting}
}
The process of state spectrum broadcasting was first proposed in ref. \cite{Horodecki_SSB_2015}, imposing a specific form for the final state of the central system and a fragment of its environment as the reason for objectivity. It was suggested as a necessary and sufficient condition for emergence of objectivity in the quantum realm. Later, however, it was proven to be too restrictive, providing sufficient, but not necessary, conditions \cite{Le_StrongQD_2019}. Here we adapt the Brand\~{a}o, Piani and Horodecki's approach to deal with State Spectrum Broadcasting and show that noncontextuality also emerges under this process.

Consider the whole system-environment dynamics $\Phi^{A\mathcal{E}}$ as a cptp map $\mathcal{D}(\mathcal{H}_A)\mapsto \mathcal{D}(\mathcal{H}_{A\mathcal{E}})$. Note this is similar to the $\Phi$ map we consider in definition \ref{def:EWdynamics}, but with the difference that we are not discarding the central system's state after the interaction. Now, as in the case of EW$_t$ dynamics, let us discard all of the environment but a small fragment with $r$ subsystems \footnote{The use of the label $r$ for the number of subsystems, instead of $t$, shall be clear soon.}, $B_{S_r}$, i.e. $\Phi^{AB_{S_r}}:={\rm Tr_{\mathcal{E}\backslash B_{S_r}}}\circ\Phi^{A\mathcal{E}}$. Again, the essential difference between $\Phi^{AB_{S_r}}$ and $\Phi^{B_{S_t}}$ is the explicit presence of $A$. Depending on the final state induced by the $\Phi^{AB_{S_r}}$ dynamics, we may arrive at a State Spectrum Broadcasting process.

\begin{definition}[State Spectrum Broadcasting]
\label{def: SSB} Consider the map $\Phi^{A\mathcal{E}}$ describing the interaction of system $A$ and its environment $\mathcal{E}$, composed of $N$ subsystems. Now, consider the dynamics $\Phi^{AB_{S_r}}$, obtained by tracing out the whole environment but $r$ subsystems denoted by $B_{S_r}$.
A State Spectrum Broadcasting process has occurred if the final joint state of $A$ and $B_{S_r}$, $\rho^{AB_{S_r}}:=\Phi^{AB_{S_r}}(\rho^A)$, takes the following form:
\begin{align}
    \rho^{AB_{S_r}}=\sum \tilde{p}_k \ket{k}\bra{k}\otimes \bigotimes_{{j\in S_t}}\sigma^{B_j}_k,
    \label{eq: SSBstructure}
\end{align}
being $\tilde{p}_k:={\rm Tr}[\ket{k}\bra{k}\rho^A]$ the probability distribution arising due to perfect decoherence, and all $\{\sigma^{B_j}_k\}_k$ having disjoint supports, i.e. $\sigma^{B_j}_k\sigma^{B_j}_{k'}=0$ whenever $k\neq k'$. 
\end{definition}
In the above definition, the name of the process becomes clear: the spectrum of the  state of the system after decoherence (with respect to a certain pointer basis), $(\tilde{p}_k)_k$, is perfectly broadcast to the environment. The proposal that such structure for the joint state $\rho^{AB_{S_r}}$ provides sufficient conditions for objectivity comes from the following: first, the system undergoes a complete decoherence process, having no decoherence-free subspace left. This means it can be considered as a classical mixture of orthogonal states. Secondly, the condition of disjoint supports for $\{\sigma^{B_j}_k\}_k$ for every $j\in B_{S_r}$ means that the distribution $\tilde{p}_k$ is perfectly broadcast to the environmental subsystems $B_j\in B_{S_r}$. That is, this condition corresponds to the $\eta=1$ case for quantum Darwinism! Thus, if one accepts a condition $\eta\approx 1$ leading to objectivity under Brand\~{a}o, Piani and Horodecki's approach to quantum Darwinism, one is forced to accept State Spectrum Broadcasting as reaching objectivity as well. Finally, the additional presence of the (completely decohered) state of the central system ensures that one can directly probe it, thus allowing for confirmation of the information broadcast to the environment (as long as one probe it using the $\{\ket{k}\bra{k}\}_k$ measurement).

Note that we can describe a State Spectrum Broadcasting process as a special and more restrictive case of Darwinism in the Brand\~{a}o, Piani and Horodecki formalism. All we need to do is to adapt our interpretation, by considering a post-interaction state of $A$ in addition to an environmental portion, $B_{S_r}$. In other words, we can choose a special $B_{S_t}$ portion in our definition of Darwinist process: $B_{S_t}= A\cup B_{S_r}$, where $B_{S_t}$ now also contains the final state of $A$. This brings no problem, as there is no need to demand all of $B_j$ \emph{to be different from $A$}.  Indeed, all that is required by the EW$_t$-dynamics is that $\Phi$ maps $\mathcal{D}(\mathcal{H}_A)$ to $\mathcal{D}(\mathcal{H}_{B_1}\otimes\ldots\otimes\mathcal{H}_{B_N})$, with $N$ large, and we should trace out all but the $B_{S_t}$ subsystems in the end -- and one of the $B_j\in B_{S_t}$ can certainly represent the state of $A$ after the interaction\footnote{Actually, this is explicitly mentioned in ref. \cite{BPH_2015_QD}.}. In the main text, we considered $B_{S_t}$ as composed of only environmental subsystems for the sake of simplicity.

Therefore, we can see State Spectrum Broadcasting as a special kind of QD$_{\eta}$ process, in which: $i)$ $\eta=1$, $ii)$ perfect decoherence has occurred ($\tilde{E}_k=\ket{k}\bra{k}$ for some pointer basis $\ket{k}$ of the Hilbert space $\mathcal{H}_A$) and $iii)$ $A\in B_{S_t}$ and $\sigma_k^A=\ket{k}\bra{k}$. Indeed, the state in Eq.~\eqref{eq: SSBstructure} is already in a measure-and-prepare form,  with objective observable $\tilde{E}_k=\ket{k}\bra{k}$ and conditionally prepared states $\ket{k}\bra{k}\otimes\bigotimes_{j\in S_t}\sigma_k^{B_j}$. Those states \emph{have disjoint support}, thus being perfectly distinguishable, i.e. $\eta=1$. From the example of perfect decoherence, we know that $\hat{P}[(\ket{k}\bra{k})_k]=1-1/2{\rm dim}(\mathcal{H}_A)$, which means $\eta(=1)>~\hat{P}$. Since we are dealing with a specific form of a QD$_\eta$ process, we can apply our Theorem \ref{theorem: NCunderQD} and conclude that, under a process of State Spectrum Broadcasting, noncontextuality emerges. Formally:

\begin{corollary*}[State Spectrum Broadcasting -- Restatement]
If the interaction leads to the occurrence of State Spectrum Broadcasting process for arbitrary initial states $\rho^A$, all the Bobs can construct a noncontextual ontological model to their statistics.
\begin{proof}
As discussed above, we have $\eta=1$ and $\tilde{E}_k=\ket{k}\bra{k}$ for all $k$, with $\ket{k}\bra{k}$ a basis of $\mathcal{H}_A$. This last condition implies $\hat{P}=1-\frac{1}{2{\rm dim (\mathcal{H}_A)}}$, as seen in sec. \ref{App:perfectDec}. If the interaction is such that a State Spectrum Broadcasting process occurs regardless of the initial state of the system, we have a QD$_{\eta}$ process satisfying
$\eta>\hat{P}$, and, by Theorem~\ref{theorem: NCunderQD}, all of the Bobs $B_j\in B_{S_t}$ can construct a noncontextual ontological model for their statistics. Interestingly, one of the Bobs will be receiving the central system $A$ after the interaction, in state $\sigma^A=\sum \tilde{p}_k\ket{k}\bra{k}$. The fact that this is a state of $A$ and not an environmental subsystem makes no difference to our results.
\end{proof}
\end{corollary*}

{
\section{Connection to the characterization of noncontextuality in the generalized probabilistic theories framework}
\label{app:NCGPT}
}
 In Lemma ~\ref{lemma: AIimplyNC}, we have showed that noncontextuality emerges for those Bob in which the encoding states $\{\sigma^{B_j}_k\}_k$ are affinely independent. The condition of affine independence of these states can be linked to the characterization of noncontextuality in  the generalized probabilistic theories (GPT) framework, as formalized in references \cite{shahandeh2021contextuality,schmid2021characterization}. For the benefit of the readers familiar with such characterization, in this section we  concisely expose this connection for the case $k_{\rm max}={\rm dim}(\mathcal{H}_{B_j})^2$. In what follows, we use the notation $n:={\rm dim}(\mathcal{H}_{B_j})$.

As mentioned in Lemma \ref{lemma: AIimplyNC}'s proof, affine independence of $\{\sigma^{B_j}_k\}_k$ implies that ${\rm ConvHull}[\{\sigma^{B_j}_k\}_k]$ is a $(k_{\rm max}-1)$ simplex, denoted by $\triangle_{k_{\rm max}}$. Thus, the effective states arriving to Bob $B_j$, $\sum \tilde{p}_k\sigma^{B_j}_k$, can be identified with points on that simplex. In the case we focus here, i.e. $k_{\rm max}=n^2$, this simplex is written $\triangle_{n^2}$. Since these are  quantum states on Bob's subsystem, $\triangle_{n^2}\subset \mathcal{D}(\mathcal{H}_{B_j})$. 

The states $\mathcal{D}(\mathcal{H}_{B_j})$, as well as the restricted $\triangle_{n^2}$, live in (a hiperplane of) the real vector space of $n\times n$ Hermitian matrices, $\mathbb{H}_n$. The set of elements of POVM, which will assign the probabilities to each outcome of all possible measurements that Bob $B_j$ can perform, can be identified with the compact convex set  ${\mathcal{D}^*(\mathcal{H}_{B_j}):=~\{F_b\in~ \mathbb{H}_n| {\rm Tr}\{F_b \rho\}\in[0,1]\,\,\forall \rho\in\mathcal{D}(\mathcal{H}_{B_j})\}}$. We also consider the analogously defined set of Hermitian matrices whose inner product with elements of \emph{the simplex} is  between $0$ and $1$, i.e. ${\triangle_{n^2}^*:=\{F'_b\in \mathbb{H}_n| {\rm Tr}\{F'_b \rho\}\in[0,1]\forall \rho\in\triangle_{n^2}\}}$.  Since $\triangle_{n^2} \subset \mathcal{D}(\mathcal{H}_{B_j})$, the `dual' set $\triangle_{n^2}^*$ must be larger than $\mathcal{D}^*(\mathcal{H}_{B_j})$, i.e. $\mathcal{D}^*(\mathcal{H}_{B_j})\subset \triangle_{n^2}^*$. 

Now, consider a generalized probabilistic theory that is the result of the following procedure: start with quantum theory (as a GPT) to then impose the restriction $\mathcal{D}(\mathcal{H}_{B_j})\rightarrow \triangle_{n^2}$ on Bob's states due to the dynamics -- let us call this GPT the `EWQuantum' theory. Under this perspective, we can view $\triangle_{n^2}$ as the normalized state space of `EWQuantum' and $\mathcal{D}^*$ as the (restricted) set of effects of such a theory.  Under this perspective, the fact that the normalized state space is a simplex, $\triangle_{n^2}$, and the effects are a subset of the dual, $\mathcal{D}^*\subset \triangle^*_{n^2}$, imply EWQuantum `fits inside' classical probability theory! Formally, this means that this GPT is \emph{simplex-embeddable}, which is an equivalent condition for noncontextuality of the underlying operational theory in prepare and measure scenarios, as proven in ref.~\cite{schmid2021characterization}. 

In ref. \cite{shahandeh2021contextuality}, noncontextuality in GPTs is characterized in a similar way,  but with the additional condition that
the classical theory in which the GPT `fits' into should
have the same dimension as the GPT. This is exactly what happens here, since the classical probability theory (as a GPT) with set of normalized 
states $\triangle_{n^2}$ and effects 
$\triangle^*_{n^2}$ lives in a real vector space of 
dimension $n^2$, which is the same as ${\rm dim}(\mathbb{H}_{n})$. Therefore, the conditions for noncontextuality according to ref. \cite{shahandeh2021contextuality} are also 
satisfied. 

For cases in which $k_{\rm max}<n$, the situation is not so straightforward (note that there is no need to consider $k_{\rm max}>n$, as it is impossible to have as many affinely independent states). We can still think of the restricted state space $\triangle\subset\mathcal{D}(\mathcal{H}_{B_j})$, but now $\mathbb{H}_n$ is strictly larger than ${\rm span}(\triangle_{k_{\rm max}})$. This brings some technical problems, e.g. $\triangle^*_{k_{\rm max}}$ would not be a compact set on $\mathbb{H}_n$. There are a few possibilities to try to dodge these problems; for instance, one could see the impact of considering non-tomographically complete GPTs, as introduced in ref. \cite{gitton2021solvable}. This, however, is out of the scope of this work, and will be addressed in the future. 

\end{document}